\documentclass[aps,twocolumn,floatfix,superscriptaddress,longbibliography,normalem]{revtex4-2}

\usepackage{adjustbox}
\usepackage{amssymb,amsmath,amstext,amsthm}
\DeclareFontFamily{U}{mathx}{}
\DeclareFontShape{U}{mathx}{m}{n}{<-> mathx10}{}
\DeclareSymbolFont{mathx}{U}{mathx}{m}{n}
\DeclareMathAccent{\widecheck}{0}{mathx}{"71}
\usepackage{dsfont}
\usepackage{graphicx}
\usepackage{bm}
\usepackage{algorithm,algpseudocode}
\usepackage{appendix}
\usepackage[T1]{fontenc}
\usepackage{physics}
\usepackage{xcolor}
\usepackage[colorlinks=true,citecolor=blue,linkcolor=magenta]{hyperref}
\usepackage{tikz-network}
\usepackage{orcidlink}

\newtheorem{theorem}{Theorem}
\newtheorem{definition}{Definition}
\newtheorem{statement}{Statement}
\newtheorem{construction}{Construction}

\newcommand{\quantummotionaus}{Quantum Motion, 4 Cornwallis St, Eveleigh NSW 2015, Australia}
\newcommand{\quantummotionuk}{Quantum Motion, 9 Sterling Way, London N7 9HJ, United Kingdom}
\newcommand{\unimelb}{School of Physics, The University of Melbourne, Parkville, VIC 3010, Australia}
\newcommand{\deloitte}{Deloitte Consulting, LLP}
\newcommand{\lanltdivision}{Theoretical Division, Los Alamos National Laboratory, Los Alamos, New Mexico 87545, USA}

\begin{document}

\title{Multivariate quantum state preparation with optimized tensor networks}

\author{Matthew~L.~Sims-Goh\,\orcidlink{0000-0002-7478-4026}}
\email{matt.sims-goh@quantummotion.tech}
\affiliation{\quantummotionaus}
\affiliation{\quantummotionuk}
\affiliation{\unimelb}
\author{Lukasz~Cincio\,\orcidlink{0000-0002-6758-4376}}
\affiliation{\lanltdivision}
\author{Annina~Z.~Lieberherr\,\orcidlink{0000-0002-6432-3138}}
\affiliation{\quantummotionuk}
\author{Mekena~McGrew\,\orcidlink{0000-0001-9040-2924}}
\affiliation{\deloitte}
\author{Thomas~R.~Bromley\,\orcidlink{0000-0002-7480-7478}}
\affiliation{\quantummotionuk}

\begin{abstract}
Quantics tensor trains are attracting intense interest for quantum-inspired computing and quantum state preparation. These methods, which approximate continuum functions by representing their amplitude encoding as a matrix product state (MPS), are exceedingly powerful for univariate functions but rapidly become challenging when handling multivariate functions, since the linear chain topology leads to a large distance between highly-entangled qubits. We overcome this limitation by introducing \textsc{scent} (Spectral Clustering for Entanglement miNimizing Trees).  \textsc{scent} is a protocol that utilizes efficiently-computable pairwise entanglement metrics to determine a suitable tree tensor network (TTN) structure, which can then be efficiently approximated using tensor cross-interpolation (TCI); we find that it substantially outperforms MPS methods and improves upon previous TTN methods. We then apply these optimized TTNs to state preparation, introducing an approximate circuit compilation method based on environment-tensor methods without significantly conceding overall accuracy. Importantly, the inherent gauge freedom of TTNs can be directly exploited in this method, resulting in higher fidelity at a given circuit depth. We demonstrate this quantum state-preparation pipeline on archetypal state-preparation problems in quantum chemistry and financial portfolio optimization. In our flagship demonstration, we encode a 20-variable probability distribution with long-ranged, non-nearest neighbor inter-variable correlations in a 200-qubit state-preparation circuit with infidelity $7.44\times 10^{-9}$ using only 43284 CNOTs; depth and fidelity can be traded, allowing the same distribution to be prepared to infidelity $10^{-3}$ with as few as 5584 CNOTs.
\end{abstract}
\maketitle

\section{Introduction\label{sec:introduction}}
The remarkable experimental advances in quantum computing of recent years \cite{bluvstein2023logical,reichardt2024demonstration,google2025quantum} suggest that we are rapidly approaching useful quantum advantage, spurring a renewed search for applications of quantum processors. There are clear prospects for practically or scientifically useful applications leveraging quantum algorithms for simulating quantum systems \cite{lloyd1996universal,mcardle2020quantum,bauer2020quantum,daley2022practical,huang2025fullqubit}, solving systems of linear equations~\cite{harrow2009quantum}, simulating the dynamics of high-dimensional linear~\cite{morales2024quantum} and nonlinear~\cite{tennie2025quantum} systems, performing generalized matrix arithmetic with QSVT~\cite{gilyen2019quantum}, and accelerating unstructured search tasks~\cite{grover1996fast,brassard2000quantum}. However, these proposed applications face a common threat: each relies upon the assumption that suitable input states can efficiently be prepared by some oracular process. Preparing general quantum states is exponentially costly~\cite{knill1995approximation,mottonen2004transformation}, raising the possibility that in-principle quantum advantage in these lucrative applications may be removed by state-preparation bottlenecks \cite{aaronson2015read,lee2023evaluating}. Efficient state preparation is therefore a crucial enabler for quantum advantage on most useful applications, and is typically built upon some parsimonious classical representation well-suited to the type of information to be encoded~\cite{kivlichan2018quantum,goh2023lie,berry2025rapid,ran2020encoding,rudolph2022decomposition,bohun2024entanglement,sugawara2025embedding,holmes2020efficient,endo2020quantum,garcia2021quantuminspired,mcardle2022quantum,moosa2023linear,gonzalez2024efficient,mori2024efficient,rosenkranz2025quantum,rattew2022preparing,marin2021quantum,huggins2025efficient,kitaev2008wavefunction,rattew2021efficient,markov2022generalized,iaconis2024quantum,grover2002creating,zoufal2019quantum,rupprecht2026sparse,fomichev2024initial,barligea2026enabling}.

One of the most popular embeddings for inputting data to quantum algorithms acting on $n$ qubits is the amplitude encoding, whereby an $\ell^2$-normalized, $2^n$-dimensional vector $\vec{\alpha}=(\alpha_0,\dots,\alpha_{2^{n-1}})$ is embedded in the amplitudes of a quantum state $\ket{\psi}=\sum_{j}\alpha_j \ket{j}$. In this work we focus on amplitude-encoded functions, which have a broad range of applications including financial modeling~\cite{rebentrost2018quantum, orus2019quantum, stamatopoulos2020option, stamatopoulos2022towards}, first-quantized chemistry~\cite{chan2023grid} and the simulation of linear and nonlinear systems \cite{morales2024quantum,tennie2025quantum}, such as those encountered in computational fluid dynamics~\cite{gourianov2024tensor,lee2026quantum}. In the standard function encoding \cite{zalka1998simulating,grover2002creating} now commonly known as the `quantics' representation in tensor-network literature \cite{khoromskij2011d,khoromskij2014tensor,shinaoka2023multiscale,ritter2024quantics,tindall2024compressing,jolly2025tensorized,nunez2025learning,niedermeier2025solving,kim2025strong,ishida2025low}, a univariate function $f:[0,1)\to \mathbb{C}$ is embedded in the state
\begin{equation}
    \ket{\psi}=\sum_j \frac{f(x_j)}{\mathcal{Z}}\ket{j}, \quad x_j\equiv \sum_{b=1}^{n}\frac{\sigma_{jb}}{2^b} = \underbrace{0.\sigma_{j1} \sigma_{j2} \dots \sigma_{jn}}_{\text{binary representation}},\label{eqn:quantics_representation}
\end{equation}
where $\mathcal{Z}^{2} = \sum_{j} |f(x_{j})|^{2}$ is chosen to $\ell^2$-normalize $\ket{\psi}$ and $\ket{j} \equiv \ket{\sigma_{j1} \sigma_{j2} \dots \sigma_{jn}}$ indexes the discretized spatial points in binary counting order.
This discretizes the variable onto an exponentially dense grid with spacing $\delta=2^{-n}$. A generic domain $y\in[a,b]$ can be trivially accommodated by function rescaling, and multivariate functions can in principle be encoded by entangling multiple registers of this type. In practice, the state preparation of multivariate functions has proven highly challenging~\cite{kitaev2008wavefunction,garcia2021quantuminspired,rattew2022preparing,moosa2023linear,mori2024efficient,manabe2025state,rosenkranz2025quantum}.

\begin{figure*}
    \centering
    \includegraphics[width=\textwidth]{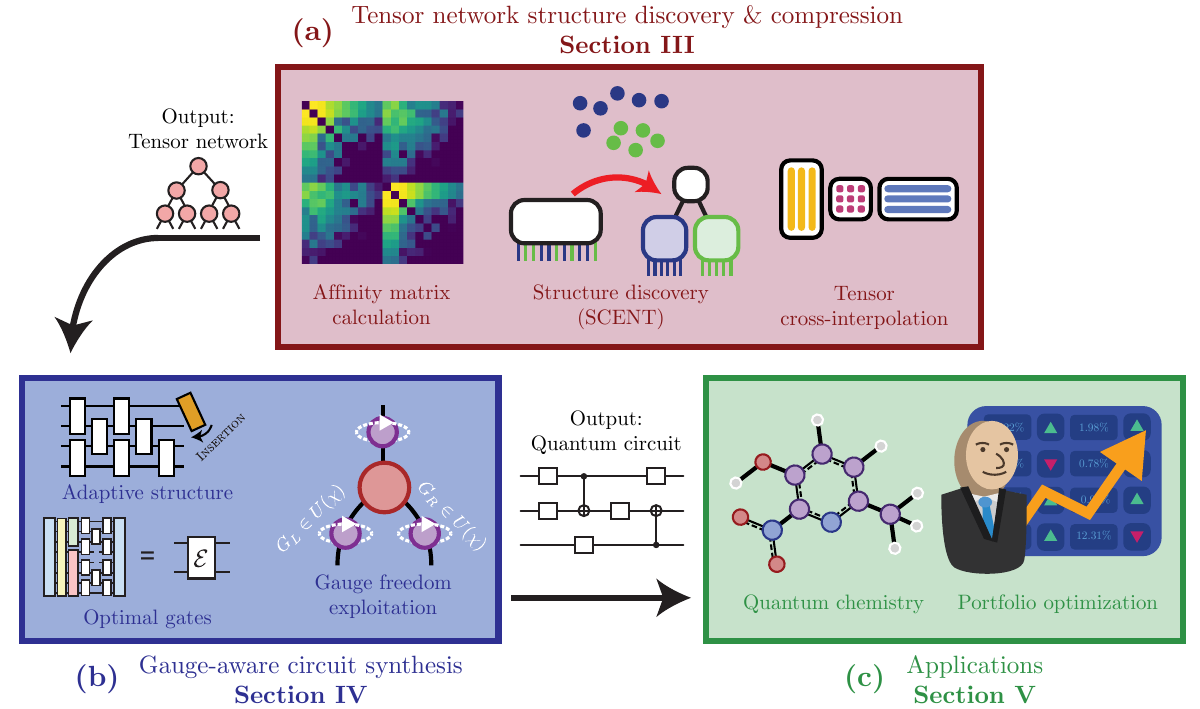}
    \caption{\label{fig:summary}\textbf{Summary of our approach.} \textbf{(a)} Our approach to tensor network structure discovery (Section~\ref{sec:ttn_compression}) begins with the calculation of affinity matrices (left), which quantify the correlation between qubits for the desired state. These are used to recursively partition qubits (middle), assigning qubit indices to branches in a way that seeks to minimize graph distance between highly entangled qubits. This tree structure is then employed in tensor cross-interpolation (right), yielding an interpolative tensor network approximation of the desired function. \textbf{(b)} We demonstrate how to synthesize a quantum circuit from this tensor network (Section~\ref{sec:compilation}). Our approach adaptively modifies circuit structure, exactly computes locally-optimal gate replacements, and exploits the inherent gauge freedom of tensor networks to construct a suitable quantum circuit. \textbf{(c)} We explore applications of multivariate function preparation on quantum computers (Section~\ref{sec:applications}), including quantum chemistry and financial portfolio application.}
\end{figure*}

A variety of distinct methods have been proposed for quantum state preparation of functions in the quantics representation \eqref{eqn:quantics_representation}, each targeting different classes of functions limited by different constraints. These include functions with sparse Fourier, polynomial or Chebyshev representations~\cite{holmes2020efficient,endo2020quantum,garcia2021quantuminspired,mcardle2022quantum,moosa2023linear,gonzalez2024efficient,mori2024efficient,rosenkranz2025quantum}, smooth univariate probability distributions~\cite{grover2002creating,zoufal2019quantum} and complex-valued functions \cite{rattew2022preparing,marin2021quantum,bohun2024entanglement}, normal distributions~\cite{kitaev2008wavefunction,rattew2021efficient,markov2022generalized,iaconis2024quantum,manabe2025state}, and Gaussian-type basis functions for quantum chemistry~\cite{huggins2025efficient}.

The methodology of the state-preparation routines mentioned above is extremely diverse, ranging from random walks~\cite{rattew2021efficient} to the quantum singular value transformation~\cite{mcardle2022quantum}. In this work, we focus on tensor network methods as a particularly promising avenue for further development. In the case of matrix product states (MPS), the maturity of density matrix renormalization group (DMRG) techniques~\cite{white1992density,schollwock2005density,schollwock2011density} coupled with a broad range of procedures to compile state-preparation circuits for MPS \cite{schon2005sequential,iten2016quantum,ran2020encoding,rudolph2022decomposition,gundlapalli2022deterministic,melnikov2023quantum,malz2024preparation,berry2025rapid} enables high-fidelity state preparation for some high-value quantum simulation problems~\cite{berry2025rapid}. Amplitude-encoded univariate functions that are sufficiently smooth or have few sharp discontinuities have provably low bond dimension when approximated by an MPS \cite{holmes2020efficient,ali2023approximation,bohun2024entanglement,holmes2020entanglement,garcia2021quantuminspired,marin2021quantum,lindsey2023multiscale}, raising the possibility of similar successes for state preparation of multivariate functions~\cite{ali2024multivariate}.

Until recently, due to the lack of a direct rank-revealing construction for MPS function approximations, MPS-based state-preparation routines were limited in their approach, relying upon expansion into analytically-constructible polynomial terms~\cite{holmes2020efficient,iaconis2024quantum}. In the past few years, however, there has been a flurry of interest in tensor cross-interpolation (TCI)~\cite{nunez2025learning}, a rank-revealing, active-learning algorithm to efficiently construct MPS via few queries to the full tensor (a generalized form of which is detailed in Section~\ref{sec:tci_intro}). TCI has enabled powerful quantum-inspired approaches to numerous problems including option pricing~\cite{sakurai2025learning}, nonlinear dynamics \cite{gourianov2024tensor,niedermeier2025solving}, and the evaluation of Feynman diagrams~\cite{nunez2022learning,ishida2025low}. TCI is now often combined with the quantics representation~\eqref{eqn:quantics_representation} for high-precision, highly-compressible representation of functions \cite{shinaoka2023multiscale,ritter2024quantics,tindall2024compressing,jolly2025tensorized,nunez2025learning,niedermeier2025solving,kim2025strong,ishida2025low}, a capability which has recently been applied to quantum state preparation~\cite{melnikov2023quantum,bohun2024entanglement}.

Although MPS can be used to represent multivariate functions via the `interleaved' and `serial' representations \cite{rodriguez2024chebyshev,ritter2024quantics,nunez2025learning,tindall2024compressing} (Section~\ref{sec:mps_qtt_intro}), their 1D topology is ill-suited to this, particularly in the presence of many variables and/or strong inter-variable correlations --- in general, the bond dimension may scale exponentially in the number of variables. Substantial advances have been made in tensor-network models of physical systems by considering more general network topologies that are better suited to the underlying entanglement structure \cite{shi2006classical,vidal2007entanglement,verstraete2004renormalization}, and it stands to reason that similar advances should be possible for the encoding of functions. Of the various generalizations of MPS, tree tensor networks (TTNs)~\cite{shi2006classical} seem a particularly appealing choice, since they can represent more general correlation structures, their loop-free structure enables efficient contractions, and they can easily be placed in useful gauges \cite{shi2006classical,tindall2023gauging,tindall2024compressing}.

Recent work by Tindall \emph{et al.} has generalized TCI to TTNs~\cite{tindall2024compressing}, showing that TTNs with heuristically chosen network structure outperform MPS in parsimoniously approximating some multivariate functions. 
These heuristically-chosen structures are now being applied to quantum state preparation~\cite{ballarin2025efficient}, but network structures need to be chosen with care, and different states will have different optimal network structures. This creates a pressing need for systematic approaches to constructing TTNs of efficient topology. While local-update-based structural optimization techniques \cite{hikihara2023automatic,manabe2025state,hikihara2025improving} can provide some assistance in this regard, they require one to first construct an MPS to `unfold', which is not guaranteed to have tractable bond dimension even if the final TTN does. 

In response to these needs, we introduce a framework for state preparation of multivariate functions, depicted schematically in Figure~\ref{fig:summary}. Specifically, we make the following contributions:
after a brief review of the relevant tensor network tools (Section~\ref{sec:preliminaries}), Section~\ref{sec:ttn_compression} introduces \textsc{scent} (Spectral Clustering for Entanglement miNimizing Trees), a novel and scalable approach to discovering an optimal TTN structure. Since it does not require an unstructured initial efficient representation of the underlying state (e.g. an MPS), it can be applied to functions too complicated to initialize for existing methods. We provide a Python implementation of these tools in the GitLab repository of Ref.~\cite{gitlab_repo}.
Section~\ref{sec:compilation} subsequently describes a new method for approximately compiling isometries --- the building blocks of the TTNs found by \textsc{scent} --- to quantum circuits.
Finally, we demonstrate these new capabilities on a range of benchmarking problems, as well as archetypal state-preparation problems in quantum chemistry and financial portfolio optimization in Section~\ref{sec:applications}.

We note additionally that tensor-network-based treatments of functions and state preparation are extremely active fields of research at present, and we present a detailed comparison to related works in the literature in Appendix~\ref{app:literature_comparison}.

\section{Preliminaries\label{sec:preliminaries}}
\subsection{Tensor networks}
Tensor networks have proven to be an extremely powerful tool for parsimonious approximation of otherwise exponentially-sized mathematical objects describing quantum systems. For the purposes of this work, a tensor network is defined as follows:
\begin{definition}[Tensor network]
    A tensor network is a collection of tensors $\{T_v\}_{v\in V}$, each with a finite set of indices, together with a specification of which indices are to be pairwise identified and summed over. This specification defines a network graph $G=(V,E)$, where vertices correspond to tensors and each edge $(i,j)\in E$ corresponds to a shared index contracted between $T_i$ and $T_j$.\label{defn:tensor_network}
\end{definition}
The summed-over indices are commonly known as internal indices or bonds, while the open indices $\bm{\sigma}=[\sigma_1,\dots,\sigma_n]$ are referred to as external or physical indices. For the purposes of this work, these correspond to the indices of each qubit in an $n$-qubit Hilbert space. The overall value of the network is the multidimensional array $\mathcal{T}(\bm{\sigma})$ obtained by performing all internal contractions, depending only on external indices $\bm{\sigma}$. The complexity of a tensor network is commonly quantified by its bond dimension $\chi$, the maximum dimension of the internal indices.
\subsection{Matrix product states and quantics tensor trains}
\label{sec:mps_qtt_intro}
The simplest, and most popular tensor network for representing quantum systems is the aforementioned MPS:
\begin{definition}[Matrix product state]
    A matrix product state (MPS) is a tensor network (Definition~\ref{defn:tensor_network}) whose underlying graph $G$ is a path graph (i.e. forms a linear chain structure).
\end{definition}

As noted above, we consider $n$-qubit states whose indices are labeled $\bm{\sigma}=[\sigma_1,\dots,\sigma_n]$ with $\sigma_i \in \{0, 1\}$. The amplitudes $\Psi_{\bm{\sigma}}=\bra{\sigma_1\dots\sigma_n}\ket{\psi}$ of such a quantum state $\ket{\psi}$ can be written as the MPS
\begin{equation}
\Psi_{\bm{\sigma}}=\prod_{b=1}^n T_b^{\sigma_b}=[T_1]^{\sigma_1}_{\alpha_1}[T_2]^{\sigma_2}_{\alpha_1\alpha_2}\dots [T_{n-1}]^{\sigma_{n-1}}_{\alpha_{n-2}\alpha_{n-1}} [T_n]^{\sigma_n}_{\alpha_{n-1}},
\label{eqn:mps_equation_form}
\end{equation}
where Einstein summation over repeated internal indices $\alpha_b$ is assumed. This is classically efficient for states of 1D systems with area-law entanglement~\cite{eisert2010colloquium}. Throughout this work, we use standard diagrammatic notation for tensor networks \cite{penrose1971applications,bridgeman2017hand}, representing tensors as nodes connected by lines; lines connecting two nodes correspond to bonds, whilst `external' lines with a free end correspond to physical indices. In this notation, an MPS is written as
\begin{equation}\Psi_{\bm{\sigma}}=
\begin{array}{c}
\begin{tikzpicture}
\Vertex[x=0,y=0,label=$T_1$,position=above,size=0.25,color=blue!30]{A1}
\Vertex[x=1,y=0,label=$T_2$,position=above,size=0.25,color=blue!30]{A2}
\Vertex[x=2.5,y=0,label=$T_{n-1}$,position=above,size=0.25,color=blue!30]{Anm1}
\Vertex[x=3.5,y=0,label=$T_n$,position=above,size=0.25,color=blue!30]{An}
\draw[line width = 1] (A1) -- ++(0,-0.35) node[below] {\scriptsize$\sigma_1$};
\draw[line width = 1] (A2) -- ++(0,-0.35) node[below] {\scriptsize$\sigma_2$};
\draw[line width = 1] (Anm1) -- ++(0,-0.35) node[below] {\scriptsize$\sigma_{n-1}$};
\draw[line width = 1] (An) -- ++(0,-0.35) node[below] {\scriptsize$\sigma_n$};
\Edge[lw=1,position={below},label={\makebox[1em][c]{$\alpha_{n-1}$}}](Anm1)(An)
\Edge[lw=1,position={below},label={\makebox[1em][c]{$\alpha_1$}}](A1)(A2)
\Edge[lw=1,label=$\cdots$](A2)(Anm1)
\end{tikzpicture}
\end{array}.
\label{eqn:mps_diagram_form}
\end{equation}
When used to represent a quantics function~\eqref{eqn:quantics_representation}, where each qubit represents a binary bit of significance in a function variable, an MPS is sometimes referred to as a quantics tensor train (QTT); we use these terms interchangeably in this work.  The multivariate case is more complicated: consider a $D$-dimensional multivariate function $f(\vec{x})$, where $\vec{x}=(x^{(1)},x^{(2)},\dots,x^{(D)})$ with $x^{(d)}\in[0,1)^B$ is a vector-valued argument. This can be encoded in a quantum state analogously to \eqref{eqn:quantics_representation} as
\begin{equation}
    \ket{\psi} = \sum_{\bm{j}}\frac{f(\vec{x}_{\bm{j}})}{\mathcal{Z}}\ket{\bm{j}}, \quad \ket{\bm{j}}=\bigotimes_{k=1}^D \ket{j_k},\label{eqn:multivariate_quantics}
\end{equation}
where $x_b^{(d)}$ is the $b$-th bit of precision in the binary expansion of the $d$-th function variable $x^{(d)}$ and the individual $\ket{j_k}$ follow the form in Equation~\ref{eqn:quantics_representation}. For univariate functions, the assignment of qubits to external tensor indices is made in order of bit significance. Limited accommodation of multivariate functions is possible by the interleaved and serial representations~\cite{ritter2024quantics,nunez2025learning,tindall2024compressing}, where in the MPS we have bit orderings
\begin{equation}
    \bm{\sigma} = \begin{cases}
        (x^{(1)}_1 \dots x^{(D)}_1 x^{(1)}_2 \dots x^{(D)}_2 \dots ) & \text{interleaved}, \\
        (x^{(1)}_1 \dots x^{(1)}_B x^{(2)}_1 \dots x^{(2)}_B \dots)& \text{serial},
    \end{cases}
\end{equation}
respectively, where $B$ bits of precision are used for each variable; this is depicted graphically in Figure~\ref{fig:tensor_network_diagrams}(a). Although the interleaved representation is typically more performant for functions with strong inter-variable correlations, generically both representations are poorly suited to the linear chain topology of a QTT for larger $D$ (Figure~\ref{fig:tensor_network_diagrams}(b)). Both can lead to large distances between bits that are often strongly correlated. In the interleaved representation, sequential bits of the same variable are separated by $D$ bonds, which can cause scaling of bond dimension exponential in $D$ due to bond crossings between different function variables. In the serial representation, the highest-significance bits of neighboring variables are separated by $B$ bonds, and those of the most distant variables are separated by $DB$ bonds, leading to very poor scaling for functions with strong inter-variable correlations. Clearly, there is a pressing motivation to explore more problem-tailored network topologies.
\begin{figure}
    \centering
    \includegraphics[width=0.5\textwidth]{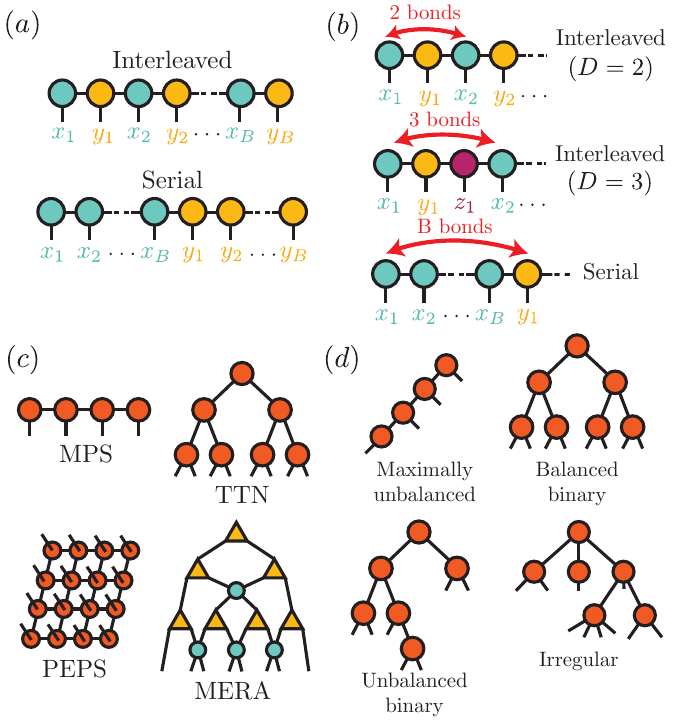}
    \caption{\label{fig:tensor_network_diagrams}\textbf{Tensor networks for function encoding.} \textbf{(a)} The interleaved (top) and serial (bottom) representations, depicted for $D=2$ dimensions.
    \textbf{(b)} Distance between closely correlated bits of the same variable. We depict how sequential bits of the same variable are separated by $D$ bonds in the interleaved representation, while the serial representation will separate the highest-significance bits of two adjacent variables by $B$ bonds. In both cases, this leads to many bond crossings that may be avoided by more problem-tailored network topologies.
    \textbf{(c)} Examples of the diverse range of tensor network structures common in the literature. Pictured are MPS (linear chain structure), TTN (acyclic structure), PEPS (Projected Entangled Pair State; loops permitted in graph structure), and MERA (Multi-scale Entanglement Renormalization Ansatz; balanced TTN with physically-motivated unitary disentanglers).
    \textbf{(d)} Examples of different TTN structures. An MPS is a maximally unbalanced TTN (top left). While balanced binary trees (top right) are common in other works, \textsc{scent} produces binary trees that may in general be unbalanced (bottom left). Furthermore, irregular trees with vertices of different degree can also be used to build valid TTNs (bottom right).}
\end{figure}

\subsection{Tree tensor networks\label{sec:ttn_methods}}
A range of more general tensor network topologies exist (see examples in Figure~\ref{fig:tensor_network_diagrams}(c)). We focus our attention on TTNs~\cite{shi2006classical}:
\begin{definition}[Tree tensor network]
    A tree tensor network (TTN) is a tensor network (Definition~\ref{defn:tensor_network}) whose underlying graph $G$ is acyclic (i.e. forms a tree-like structure).
\end{definition}
These are suitable for our purposes due to their loop-free structure; it remains unclear how TCI could be generalized to tensor networks with looped topology. TTNs encompass any tensor network where the bonds do not form loops --- various types are depicted in Figure~\ref{fig:tensor_network_diagrams}(d). 
The tensors with open (or physical) indices are also referred to as leaf tensors.

Like all tensor networks, TTNs have gauge freedom that can be exploited to expose useful properties and simplify computations. Given any two tensors $A$ and $B$ connected by a bond of dimension $\chi$, we note that
\begin{equation}
\begin{array}{c}
\begin{tikzpicture}
\Vertex[x=0,y=0,label=$A$,position=above,size=0.3,color=blue!30]{A}
\Vertex[x=0.6,y=0,label=$B$,position=above,size=0.3,color=blue!30]{B}
\draw[line width = 1] (A) -- ++(0,-0.35) node[below] {\scriptsize$\bm{\nu}_A$};
\draw[line width = 1] (B) -- ++(0,-0.35) node[below] {\scriptsize$\bm{\nu}_B$};
\Edge[lw=1](A)(B)
\end{tikzpicture}
\end{array}
=
\begin{array}{c}
\begin{tikzpicture}
\Vertex[x=0,y=0,label=$A$,position=above,size=0.3,color=blue!30]{A}
\Vertex[x=0.5,y=0,label=$X$,position=above,size=0.2,color=red!30]{X}
\Vertex[x=1.0,y=0,label=$X^\dagger$,position=above,size=0.2,color=red!30]{Xdag}
\Vertex[x=1.5,y=0,label=$B$,position=above,size=0.3,color=blue!30]{B}
\draw[line width = 1] (A) -- ++(0,-0.35) node[below] {\scriptsize$\bm{\nu}_A$};
\draw[line width = 1] (B) -- ++(0,-0.35) node[below] {\scriptsize$\bm{\nu}_B$};
\Edge[lw=1](A)(B)
\end{tikzpicture}
\end{array}
=
\begin{array}{c}
\begin{tikzpicture}
\Vertex[x=0,y=0,label=$\tilde{A}$,position=above,size=0.3,color=green!30]{A}
\Vertex[x=0.6,y=0,label=$\tilde{B}$,position=above,size=0.3,color=green!30]{B}
\draw[line width = 1] (A) -- ++(0,-0.35) node[below] {\scriptsize$\bm{\nu}_A$};
\draw[line width = 1] (B) -- ++(0,-0.35) node[below] {\scriptsize$\bm{\nu}_B$};
\Edge[lw=1](A)(B)
\end{tikzpicture}
\end{array}
\end{equation}
for any $\chi\times\chi$ unitary $X\in U(\chi)$, where we have performed the contractions $\tilde{A}=AX$ and $\tilde{B}=BX^\dagger$, and $\bm{\nu}_A$ and $\bm{\nu}_B$ are multi-indices collating all other bond and physical indices of $A$ and $B$, thus creating a gauge freedom whereby different sets of tensors can represent the same underlying state. In this work, we make use of the standard isometric gauge, detailed below, and the interpolative gauge (Section~\ref{sec:tci_intro}). First, recall the following definition:
\begin{definition}[Isometry]
Consider two Hilbert spaces $\mathcal{H}_{\operatorname{in}}$ and $\mathcal{H}_{\operatorname{out}}$, with dimensions $\chi_{\operatorname{in}}\equiv\dim(\mathcal{H}_{\operatorname{in}})$ and $\chi_{\operatorname{out}}\equiv\dim(\mathcal{H}_{\operatorname{out}})$, and where $\chi_{\operatorname{in}}\leq \chi_{\operatorname{out}}$. An isometry is a linear map $W:\mathcal{H}_{\operatorname{in}}\to\mathcal{H}_{\operatorname{out}}$ such that $W^\dagger W=\mathds{1}$, which we write in diagrammatic form as
\begin{equation}
\begin{array}{c}
\begin{tikzpicture}
\Vertex[x=0,y=0,label=$W^\dagger$,position=right,size=0.15,color=yellow!30,shape=semicircle,style={rotate=90}]{Wdag}
\Vertex[x=0.5,y=0,label=$W$,position=left,size=0.15,color=yellow!30,shape=semicircle,style={rotate=-90}]{W}
\draw[line width = 1] (W) -- ++(0.35,0) node[below] {};
\draw[line width = 1] (Wdag) -- ++(-0.35,0) node[below] {};
\Edge[lw=2](W)(Wdag)
\end{tikzpicture}
\end{array}
=
\begin{array}{c}
\begin{tikzpicture}
\draw[line width=1pt] (0,0) -- (1,0);
\end{tikzpicture}
\end{array},
\end{equation}
where the flat edge of the semicircle is chosen to represent the $\mathcal{H}_{\operatorname{out}}$ direction.
\end{definition}
If the Hilbert space dimensions are powers of $2$ ($\chi_{\operatorname{in}}=2^m$, $\chi_{\operatorname{out}}=2^n$), an isometry is a physically implementable map where $m$ qubits in any state are entangled with $(n-m)$ additional uninitialized qubits, with circuit complexity typically scaling as $\mathcal{O}(2^{n+m})$~\cite{iten2016quantum}.

Any tensor in a TTN can be straightforwardly replaced by an isometry (with respect to some bipartition of indices to turn the tensor into a matrix, known as matricization) without changing the overall value of the network as follows:
\begin{construction}[Isometrization of a tensor]
\label{construct:isometrization}Given two bonded tensors $A$ and $B$ in a TTN, we may replace $A\to \tilde{A}$ and $B\to\tilde{B}$ without changing the overall value of the TTN, where $\tilde{B}$ is an isometry, as follows. Matricizing $B$ and computing the singular value decomposition (SVD) $B=USV$, one obtains
\begin{equation}
\label{eqn:svd_push}
\begin{array}{c}
\begin{tikzpicture}
\Vertex[x=0.6,y=0.6,label=$A$,position=below,size=0.3,color=blue!30]{A}
\Vertex[x=0,y=0,label=$B$,position=above,size=0.3,color=blue!30]{B}
\Vertex[x=1.2,y=0,label=$C$,position=above,size=0.3,color=blue!30]{C}
\draw[line width = 1] (A) -- ++(0,0.35) node[above] {\scriptsize$\bm{\nu}_A$};
\draw[line width = 1] (B) -- ++(0,-0.35) node[below] {\scriptsize$\bm{\nu}_B$};
\draw[line width = 1] (C) -- ++(0,-0.35) node[below] {\scriptsize$\bm{\nu}_C$};
\Edge[lw=1](A)(B)
\Edge[lw=1](A)(C)
\end{tikzpicture}
\end{array}
=
\begin{array}{c}
\begin{tikzpicture}
\Vertex[x=1.8,y=0.6,label=$A$,position=below,size=0.3,color=blue!30]{A}
\Vertex[x=0,y=0,label=$U$,position=above,size=0.15,color=yellow!30,shape=semicircle]{U}
\Vertex[x=0.6,y=0,label=$S$,position=above,size=0.2,color=gray!30]{S}
\Vertex[x=1.2,y=0,label=$V$,position=right,size=0.15,color=yellow!30,shape=semicircle,style={rotate=90}]{V}
\Vertex[x=2.4,y=0,label=$C$,position=above,size=0.3,color=blue!30]{C}
\draw[line width = 1] (A) -- ++(0,0.35) node[above] {\scriptsize$\bm{\nu}_A$};
\draw[line width = 1] (U) -- ++(0,-0.35) node[below] {\scriptsize$\bm{\nu}_B$};
\draw[line width = 1] (C) -- ++(0,-0.35) node[below] {\scriptsize$\bm{\nu}_C$};
\Edge[lw=1](A)(V)
\Edge[lw=1](A)(C)
\Edge[lw=1](U)(S)
\Edge[lw=1](V)(S)
\end{tikzpicture}
\end{array}
=
\begin{array}{c}
\begin{tikzpicture}
\Vertex[x=0.6,y=0.6,label=$\tilde{A}$,position=below,size=0.3,color=green!30]{A}
\Vertex[x=0,y=0,label=$\tilde{B}$,position=above,size=0.15,color=green!30,shape=semicircle]{B}
\Vertex[x=1.2,y=0,label=$C$,position=above,size=0.3,color=blue!30]{C}
\draw[line width = 1] (A) -- ++(0,0.35) node[above] {\scriptsize$\bm{\nu}_A$};
\draw[line width = 1] (B) -- ++(0,-0.35) node[below] {\scriptsize$\bm{\nu}_B$};
\draw[line width = 1] (C) -- ++(0,-0.35) node[below] {\scriptsize$\bm{\nu}_C$};
\Edge[lw=1](B)(A)
\Edge[lw=1](A)(C)
\end{tikzpicture}
\end{array},
\end{equation}
where the parent tensor $A$ is replaced by contracting upwards $\tilde{A}=SVA$, and the child tensor $B$ is replaced by the re-labeled isometry $\tilde{B}=U$.
\end{construction}
This enables one to easily place a TTN in isometric gauge:
\begin{definition}[Isometric gauge]
    A tree tensor network is in \emph{isometric gauge} with respect to a root tensor known as the \emph{orthogonality center} if all other tensors are isometries with respect to it.
\end{definition}
Any TTN can be placed in isometric gauge with respect to any choice of orthogonality center by starting at the leaf tensors and `pushing up' with Construction~\ref{construct:isometrization} towards the center, in such an order that no tensor is isometrized before its children. This creates an obvious path to preparation of a quantum state represented by a TTN, since isometries are physical processes that can be decomposed into circuits. Due to the normalization of isometries $W^\dagger W=\mathds{1}$, for state-preparation purposes the TTN can be efficiently renormalized to an $\ell^2$-normalized quantum state by simply Frobenius-normalizing the orthogonality centre. The orthogonality center itself corresponds to a smaller state-preparation problem, or can alternatively be thought of as an isometry for which $\mathcal{H}_{\operatorname{in}}$ is the trivial Hilbert space.

\subsection{Tensor cross-interpolation\label{sec:tci_intro}}
In this work, a rank-revealing construction of TTN function representation is achieved with TCI, which we review here. Our implementation mostly follows the TTN generalization of Ref.~\cite{tindall2024compressing}, and a comprehensive overview from an MPS perspective can be found in Ref.~\cite{nunez2025learning}. To understand this formulation of TCI, we need to first recall the interpolative decomposition (ID) of a matrix:
\begin{definition}[Interpolative decomposition]
    Let $A$ be an $m\times n$ matrix of rank $k$. The (row) interpolative decomposition of $A$ is a factorization $A=XR$ into an $m\times k$ `interpolation matrix' $X$ and a $k\times n$ `skeleton matrix' $R$. Specifically, the skeleton matrix consists of $k$ rows chosen from $A$, denoted $R=A(\mathcal{I}_k,\mathcal{J})$, where $\mathcal{J}$ is the set of all column indices of $A$, and $\mathcal{I}_k$ is a subset of $k$ of the row indices of $A$. The interpolation matrix $X$ consists of a $k\times k$ identity matrix augmented with additional structure that interpolates between the rows $\mathcal{I}_k$.\label{defn:exact_interpolative_decomposition}
\end{definition}
Although the ID always exists, the exact rank $k$ may be large even for highly compressible matrices, thus we make frequent use of the approximate ID:
\begin{statement}[Approximate interpolative decomposition]
\label{stat:approx_interp_decomp}Let $A$ be an $m\times n$ matrix of rank $l\geq k$. The approximate (row) interpolative decomposition of $A$ is a factorization $A\approx XR$ following the structure outlined in Definition~\ref{defn:exact_interpolative_decomposition}. The error $\epsilon_{\operatorname{id}}\equiv\|A-XR\|$ is of the order $\epsilon_{\operatorname{id}}\in\mathcal{O}(\sigma_{k+1})$, where $\sigma_{k+1}$ is the $(k+1)$-th largest singular value of $A$.
\end{statement}
Constructive algorithms to efficiently compute an approximate ID are well known~\cite{nunez2025learning,cheng2005compression,liberty2007randomized}. Similar to Construction~\ref{construct:isometrization}, one can use an (approximate) ID to `interpolatize' any tensor in a TTN:
\begin{construction}[Interpolatization of a tensor]
    \label{construct:interpolatization}Given two bonded tensors $A$ and $B$ in a TTN, we may replace $A\to \tilde{A}$ and $B\to\tilde{B}$ without changing the overall value of the TTN, where $\tilde{B}$ is an interpolation matrix, as follows. Performing an ID $B=XR$ and contracting, one obtains
    \begin{equation}
\label{eqn:interpolative_push}
\begin{array}{c}
\begin{tikzpicture}
\Vertex[x=0.6,y=0.6,label=$A$,position=below,size=0.3,color=blue!30]{A}
\Vertex[x=0,y=0,label=$B$,position=above,size=0.3,color=blue!30]{B}
\Vertex[x=1.2,y=0,label=$C$,position=above,size=0.3,color=blue!30]{C}
\draw[line width = 1] (A) -- ++(0,0.35) node[above] {\scriptsize$\bm{\nu}_A$};
\draw[line width = 1] (B) -- ++(0,-0.35) node[below] {\scriptsize$\bm{\nu}_B$};
\draw[line width = 1] (C) -- ++(0,-0.35) node[below] {\scriptsize$\bm{\nu}_C$};
\Edge[lw=1](A)(B)
\Edge[lw=1](A)(C)
\end{tikzpicture}
\end{array}
=
\begin{array}{c}
\begin{tikzpicture}
\Vertex[x=1.4,y=0.6,label=$A$,position=below,size=0.3,color=blue!30]{A}
\Vertex[x=0,y=0,label=$X$,position=above,size=0.3,color=yellow!30]{X}
\Vertex[x=0.8,y=0,label=$R$,position=above,size=0.3,color=red!30]{R}
\Vertex[x=2,y=0,label=$C$,position=above,size=0.3,color=blue!30]{C}
\draw[line width = 1] (A) -- ++(0,0.35) node[above] {\scriptsize$\bm{\nu}_A$};
\draw[line width = 1] (X) -- ++(0,-0.35) node[below] {\scriptsize$\bm{\nu}_B$};
\draw[line width = 1] (C) -- ++(0,-0.35) node[below] {\scriptsize$\bm{\nu}_C$};
\Edge[lw=1](A)(R)
\Edge[lw=1](A)(C)
\Edge[Direct,lw=1](X)(R)
\end{tikzpicture}
\end{array}
=
\begin{array}{c}
\begin{tikzpicture}
\Vertex[x=0.6,y=0.6,label=$\tilde{A}$,position=below,size=0.3,color=green!30]{A}
\Vertex[x=0,y=0,label=$\tilde{B}$,position=above,size=0.3,color=green!30]{B}
\Vertex[x=1.2,y=0,label=$C$,position=above,size=0.3,color=blue!30]{C}
\draw[line width = 1] (A) -- ++(0,0.35) node[above] {\scriptsize$\bm{\nu}_A$};
\draw[line width = 1] (B) -- ++(0,-0.35) node[below] {\scriptsize$\bm{\nu}_B$};
\draw[line width = 1] (C) -- ++(0,-0.35) node[below] {\scriptsize$\bm{\nu}_C$};
\Edge[Direct,lw=1](B)(A)
\Edge[lw=1](A)(C)
\end{tikzpicture}
\end{array},
\end{equation}
by re-labelling the interpolation tensor $\tilde{B}=X$ and contracting up the sampled amplitudes $\tilde{A}=RA$. In doing so, one also obtains the indices $\mathcal{I}_k$ from the ID, informing which indices of $\bm{\nu}_B$ are exactly interpolated.
\end{construction}
The main utility of Construction~\ref{construct:interpolatization} is to place a TTN in the `interpolative gauge' introduced by Tindall \emph{et al.}~\cite{tindall2024compressing}.
\begin{definition}[Interpolative gauge]
    A tree tensor network is in \emph{interpolative gauge} with respect to a root tensor known as the \emph{interpolative center} if all other tensors are interpolation matrices (Definition~\ref{defn:exact_interpolative_decomposition}) with respect to it.
\end{definition}
Analogously to placing a TTN in isometric gauge, we can place a TTN in interpolative gauge by starting at the leaf tensors and `pushing up' with Construction~\ref{construct:interpolatization}. TTNs in interpolative gauge have the crucial property that the interpolative center consists entirely of exactly sampled amplitudes from the full $2^n$-dimensional tensor (i.e. for a quantics TTN, function values sampled at some points)~\cite{tindall2024compressing}, with the other tensors interpolating between these. The external index configurations $\bm{\sigma}$ at which these amplitudes are sampled are easily determined --- when one pushes up from a leaf tensor, the indices $\mathcal{I}_k$ from the ID reveal the exact configurations of the Hilbert space indices that correspond to each row. Therefore, unlike typical tensor network gauges the bond indices have physical meaning, and arrows in Equation~\ref{eqn:interpolative_push} indicate that the bond index is augmented with this knowledge, pointing in the direction of where those row values are located in the tensor network. Therefore, when placing the TTN in interpolative gauge, one takes the tensor product of child indices $\mathcal{I}_k$ to determine the parent index, then reduces it in rank by applying Construction~\ref{construct:interpolatization}, thereby selecting the subset of index configurations to be passed up the tree. Each time Construction~\ref{construct:interpolatization} is applied, one augments the parent bond (i.e. $(\tilde{B},\tilde{A})$ in Equation~\ref{eqn:interpolative_push}) with this information. The net result is that the interpolative center has purely incoming bonds, each augmented with information about the associated configurations of external Hilbert space indices along the corresponding branch. This is closely related to the `nesting conditions' encountered in TCI with QTTs~\cite{nunez2025learning}.

With these structures in play, we are able to proceed to TCI. Given a quantics function $f(\bm{\sigma})$ which we have efficient query access to, we seek to construct a low-bond-dimension tensor network $\mathcal{T}(\bm{\sigma})$ of fixed structure that approximates it as $\mathcal{T}(\bm{\sigma})\approx f(\bm{\sigma})$ using few queries to $f(\bm{\sigma})$. The function $f(\bm{\sigma})$ may also be thought of as a `full tensor' of exponential size, since it has values defined for all $2^n$ possible index configurations $\bm{\sigma}$. As first outlined in Ref.~\cite{tindall2024compressing}, the interpolative property described above allows one to implement TCI for TTNs, which works by `sweeping' through the bonds, moving the interpolative center and refining the approximation at each bond. We use a slightly different procedure to Ref.~\cite{tindall2024compressing}, which is essentially a TTN generalization of the 2-site, reset-mode TCI algorithm outlined in Ref.~\cite{nunez2025learning}. First, we note that while the interpolative center has values \emph{exactly} equal to the function values at the index configurations carried by its incoming bonds, this interpolative property only holds \emph{approximately} for the 2-site tensor obtained by contracting the interpolative center with its neighbor. Diagrammatically,
\begin{equation}
\label{eqn:2_site_pi_approx}
\begin{array}{c}
\begin{tikzpicture}
\Vertex[x=0,y=0,label=$A$,position=above,size=0.3,color=blue!30]{A}
\Vertex[x=0.7,y=0,label=$B$,position=above,size=0.3,color=blue!30]{B}
\draw[line width = 1,{Latex[length=2mm, width=1.5mm]}-] (A) -- ++(-0.35,0.35) node[left] {\scriptsize$\bm{\nu}_\alpha$};
\draw[line width = 1,{Latex[length=2mm, width=1.5mm]}-] (A) -- ++(-0.35,-0.35) node[left] {\scriptsize$\bm{\nu}_\beta$};
\draw[line width = 1,{Latex[length=2mm, width=1.5mm]}-] (B) -- ++(0.35,0.35) node[right] {\scriptsize$\bm{\nu}_\gamma$};
\draw[line width = 1,{Latex[length=2mm, width=1.5mm]}-] (B) -- ++(0.35,-0.35) node[right] {\scriptsize$\bm{\nu}_\delta$};
\Edge[Direct,lw=1](B)(A)
\end{tikzpicture}
\end{array}
\approx
\begin{array}{c}
\begin{tikzpicture}
\Vertex[x=0,y=0,label=$\Pi$,position=above,size=0.5,color=gray!30]{Pi}
\draw[line width = 1,{Latex[length=2mm, width=1.5mm]}-] (Pi) -- ++(-0.4,0.4) node[left] {\scriptsize$\bm{\nu}_\alpha$};
\draw[line width = 1,{Latex[length=2mm, width=1.5mm]}-] (Pi) -- ++(-0.4,-0.4) node[left] {\scriptsize$\bm{\nu}_\beta$};
\draw[line width = 1,{Latex[length=2mm, width=1.5mm]}-] (Pi) -- ++(0.4,0.4) node[right] {\scriptsize$\bm{\nu}_\gamma$};
\draw[line width = 1,{Latex[length=2mm, width=1.5mm]}-] (Pi) -- ++(0.4,-0.4) node[right] {\scriptsize$\bm{\nu}_\delta$};
\end{tikzpicture}
\end{array},
\end{equation}
where $\begin{array}{c}
\begin{tikzpicture}
\Vertex[x=0,y=0,label=$\Pi$,position=left,size=0.3,color=gray!30]{Pi}
\draw[line width = 1,{Latex[length=1mm, width=1.5mm]}-] (Pi) -- ++(-0.25,0.25) node[left] {};
\draw[line width = 1,{Latex[length=1mm, width=1.5mm]}-] (Pi) -- ++(-0.25,-0.25) node[left] {};
\draw[line width = 1,{Latex[length=1mm, width=1.5mm]}-] (Pi) -- ++(0.25,0.25) node[right] {};
\draw[line width = 1,{Latex[length=1mm, width=1.5mm]}-] (Pi) -- ++(0.25,-0.25) node[right] {};
\end{tikzpicture}
\end{array}$ is the tensor obtained by querying the full function tensor $f(\bm{\sigma})$ at the index configurations carried by the 2-site composite index $\bm{\nu}_\alpha\otimes\bm{\nu}_\beta\otimes\bm{\nu}_\gamma\otimes\bm{\nu}_\delta$. This enables the following update procedure:
\begin{construction}[TCI 2-site update]
    \label{construct:tci_2site_update}Given two bonded tensors $A$ and $B$ in a TTN, where the TTN is in interpolative gauge with $A$ the interpolative center (as in Equation~\ref{eqn:2_site_pi_approx}, an update step replaces $A\to\tilde{A}$ and $B\to\tilde{B}$ as follows. First, query the function $f(\bm{\sigma})$ at the index configurations $\bm{\nu}_\alpha\otimes\bm{\nu}_\beta\otimes\bm{\nu}_\gamma\otimes\bm{\nu}_\delta$ incoming to $A$ and $B$, constructing the exact 2-site tensor $\begin{array}{c}
\begin{tikzpicture}
\Vertex[x=0,y=0,label=$\Pi$,position=left,size=0.3,color=gray!30]{Pi}
\draw[line width = 1,{Latex[length=1mm, width=1.5mm]}-] (Pi) -- ++(-0.25,0.25) node[left] {};
\draw[line width = 1,{Latex[length=1mm, width=1.5mm]}-] (Pi) -- ++(-0.25,-0.25) node[left] {};
\draw[line width = 1,{Latex[length=1mm, width=1.5mm]}-] (Pi) -- ++(0.25,0.25) node[right] {};
\draw[line width = 1,{Latex[length=1mm, width=1.5mm]}-] (Pi) -- ++(0.25,-0.25) node[right] {};
\end{tikzpicture}
\end{array}$. Then, use an approximate interpolative decomposition (Statement~\ref{stat:approx_interp_decomp}) to factorize it as
\begin{equation}
\begin{array}{c}
\begin{tikzpicture}
\Vertex[x=0,y=0,label=$\Pi$,position=above,size=0.5,color=gray!30]{Pi}
\draw[line width = 1,{Latex[length=2mm, width=1.5mm]}-] (Pi) -- ++(-0.4,0.4) node[left] {\scriptsize$\bm{\nu}_\alpha$};
\draw[line width = 1,{Latex[length=2mm, width=1.5mm]}-] (Pi) -- ++(-0.4,-0.4) node[left] {\scriptsize$\bm{\nu}_\beta$};
\draw[line width = 1,{Latex[length=2mm, width=1.5mm]}-] (Pi) -- ++(0.4,0.4) node[right] {\scriptsize$\bm{\nu}_\gamma$};
\draw[line width = 1,{Latex[length=2mm, width=1.5mm]}-] (Pi) -- ++(0.4,-0.4) node[right] {\scriptsize$\bm{\nu}_\delta$};
\end{tikzpicture}
\end{array}
\approx
\begin{array}{c}
\begin{tikzpicture}
\Vertex[x=0,y=0,label=$X$,position=above,size=0.3,color=yellow!30]{A}
\Vertex[x=0.7,y=0,label=$R$,position=above,size=0.3,color=red!30]{B}
\draw[line width = 1,{Latex[length=2mm, width=1.5mm]}-] (A) -- ++(-0.35,0.35) node[left] {\scriptsize$\bm{\nu}_\alpha$};
\draw[line width = 1,{Latex[length=2mm, width=1.5mm]}-] (A) -- ++(-0.35,-0.35) node[left] {\scriptsize$\bm{\nu}_\beta$};
\draw[line width = 1,{Latex[length=2mm, width=1.5mm]}-] (B) -- ++(0.35,0.35) node[right] {\scriptsize$\bm{\nu}_\gamma$};
\draw[line width = 1,{Latex[length=2mm, width=1.5mm]}-] (B) -- ++(0.35,-0.35) node[right] {\scriptsize$\bm{\nu}_\delta$};
\Edge[Direct,lw=1](A)(B)
\end{tikzpicture}
\end{array}
=
\begin{array}{c}
\begin{tikzpicture}
\Vertex[x=0,y=0,label=$\tilde{A}$,position=above,size=0.3,color=green!30]{A}
\Vertex[x=0.7,y=0,label=$\tilde{B}$,position=above,size=0.3,color=green!30]{B}
\draw[line width = 1,{Latex[length=2mm, width=1.5mm]}-] (A) -- ++(-0.35,0.35) node[left] {\scriptsize$\bm{\nu}_\alpha$};
\draw[line width = 1,{Latex[length=2mm, width=1.5mm]}-] (A) -- ++(-0.35,-0.35) node[left] {\scriptsize$\bm{\nu}_\beta$};
\draw[line width = 1,{Latex[length=2mm, width=1.5mm]}-] (B) -- ++(0.35,0.35) node[right] {\scriptsize$\bm{\nu}_\gamma$};
\draw[line width = 1,{Latex[length=2mm, width=1.5mm]}-] (B) -- ++(0.35,-0.35) node[right] {\scriptsize$\bm{\nu}_\delta$};
\Edge[Direct,lw=1](A)(B)
\end{tikzpicture}
\end{array},
\end{equation}
where we have relabeled $\tilde{A}=X$ and $\tilde{B}=R$. The updated TTN is in interpolative gauge, with $\tilde{B}$ the new interpolative center.
\end{construction}
Given a specified tensor network topology, we perform TCI by randomly initializing tensors, determining a `sweep plan' that visits each bond at least once, bringing the initial random TTN into interpolative gauge, and then performing a 2-site update (Construction~\ref{construct:tci_2site_update}) at each bond in the order specified by the sweep plan. A single 2-site update can change the dimension of the bond it acts on (increasing or decreasing it), which will eventually converge. This forms an active learning loop whereby the function queries in a 2-site update both improve the approximation and inform the indices to be queried at the next site.

\section{Efficient function encoding with tree tensor networks\label{sec:ttn_compression}}
Although the techniques of Section~\ref{sec:tci_intro} allow one to efficiently construct a TTN approximation of a function given a suitable network structure, they do not provide a means to discover such a suitable topology (which instead needs to be given as fixed input). Minimizing entanglement flowing through the bonds of a tensor network will lead to lower bond dimension~\cite{verstraete2006matrix,mansuroglu2025preparation}, and therefore network structures that place highly correlated Hilbert space indices far apart will lead to less compressed TTN representations. It is therefore critical to discover efficient topologies, a need that we address in this section. Recent evidence has demonstrated that multivariate quantics functions can often be more efficiently represented by a TTN than an MPS~\cite{tindall2024compressing,manabe2025state}, which can reduce the circuit complexity required for state preparation~\cite{iten2016quantum,manabe2025state}. In the event that a suitable tree structure is known \emph{a priori}, performing TCI will efficiently approximate the function as a TTN to the desired tolerance. In some cases, a well-chosen heuristic for the network structure is enough. For example, the `comb' TTN
\begin{equation}
\label{eqn:comb}
\begin{array}{c}
\begin{tikzpicture}
\Vertex[x=0,y=0,position=below,size=0.3,color=blue!30]{X1}
\Vertex[x=1.3,y=0,position=below,size=0.3,color=blue!30]{Y1}
\Vertex[x=2.8,y=0,position=below,size=0.3,color=blue!30]{Z1}
\Vertex[x=0,y=-0.65,position=below,size=0.3,color=blue!30]{X2}
\Vertex[x=1.3,y=-0.65,position=below,size=0.3,color=blue!30]{Y2}
\Vertex[x=2.8,y=-0.65,position=below,size=0.3,color=blue!30]{Z2}
\Vertex[x=0,y=-1.65,position=below,size=0.3,color=blue!30]{X3}
\Vertex[x=1.3,y=-1.65,position=below,size=0.3,color=blue!30]{Y3}
\Vertex[x=2.8,y=-1.65,position=below,size=0.3,color=blue!30]{Z3}
\Edge[lw=1](X1)(Y1)
\Edge[lw=1](X1)(X2)
\Edge[lw=1](Y1)(Y2)
\Edge[lw=1](Z1)(Z2)
\Edge[lw=1,label=$\cdots$](Y1)(Z1)
\Edge[lw=1,label=$\vdots$](X2)(X3)
\Edge[lw=1,label=$\vdots$](Y2)(Y3)
\Edge[lw=1,label=$\vdots$](Z2)(Z3)
\draw[line width = 1] (X1) -- ++(0.35,-0.3) node[below, right] {\scriptsize$x^{(1)}_1$};
\draw[line width = 1] (X2) -- ++(0.35,-0.3) node[below, right] {\scriptsize$x^{(1)}_2$};
\draw[line width = 1] (X3) -- ++(0.35,-0.3) node[below, right] {\scriptsize$x^{(1)}_B$};
\draw[line width = 1] (Y1) -- ++(0.35,-0.3) node[below, right] {\scriptsize$x^{(2)}_1$};
\draw[line width = 1] (Y2) -- ++(0.35,-0.3) node[below, right] {\scriptsize$x^{(2)}_2$};
\draw[line width = 1] (Y3) -- ++(0.35,-0.3) node[below, right] {\scriptsize$x^{(2)}_B$};
\draw[line width = 1] (Z1) -- ++(0.35,-0.3) node[below, right] {\scriptsize$x^{(D)}_1$};
\draw[line width = 1] (Z2) -- ++(0.35,-0.3) node[below, right] {\scriptsize$x^{(D)}_2$};
\draw[line width = 1] (Z3) -- ++(0.35,-0.3) node[below, right] {\scriptsize$x^{(D)}_B$};
\end{tikzpicture}
\end{array}
\end{equation}
typically outperforms an interleaved QTT~\cite{tindall2024compressing} for smooth multivariate functions. However, this cannot be optimal in general if the order of the variables $x^{(d)}$ is presupposed (i.e. the order of the teeth of the comb), since strongly correlated variables should be placed near each other to minimize bond crossings. Moreover, even with optimal variable ordering it is possible for correlations to exist between non-consecutive variables, as we demonstrate in Section~\ref{sec:compilation_results}. This motivates our development of a systematic method to identify a suitable tree structure. Importantly, as outlined in Section~\ref{sec:introduction}, our approach does not require the condition of a (possibly intractable) initial MPS, thus enabling it to operate on some functions that may be impractical for local-update-based structural optimization techniques~\cite{hikihara2023automatic,manabe2025state,hikihara2025improving}. A detailed comparison of this section's methods to the state of the art in the literature can be found in Appendix~\ref{app:literature_compression_optimization}. We have published a Python library implementing these methods in the GitLab repository of Ref.~\cite{gitlab_repo}.

\subsection{Tree structure discovery with \textsc{scent}}
To optimize a tensor network structure, one should seek to arrange external indices in a way that minimizes entanglement flowing through the bonds of the network. The key observation underlying our method is that a binary TTN can be thought of as a hierarchical sequence of bipartitionings of the qubits. Starting with the set of all qubits at the root of the tree, each tensor divides the qubits along its branch into two new branches, continuing recursively down the tree until reaching the leaves (single-qubit external indices). Our key insight is that to obtain a TTN structure that is well-suited to approximate a particular state, one can recursively minimize the entanglement between the two branches leaving each tensor. Accordingly, our approach to structural discovery is as follows:
\begin{enumerate}
    \item Choose and compute a pairwise entanglement/correlation metric between qubits.
    \item Construct an affinity matrix with suitable rescaling.
    \item Recursively perform spectral clustering to assign the qubit indices to branches, subdividing each branch into further branches.
    \item Perform TCI to approximate the function using the discovered network structure.
\end{enumerate}
This approach is referred to as \textsc{scent}. It is particularly suitable when there is hidden structure in the correlation of function variables --- in other words, when some pairs of variables are much more strongly correlated than others, and especially when this correlation forms a tree-like structure. The final step (TCI) was detailed in Section~\ref{sec:tci_intro}; below, we remark on some subtleties regarding the first three steps.

\subsubsection{Pairwise entanglement/correlation metrics\label{sec:computing_entanglement_metrics}}
Each recursive step of \textsc{scent} requires one to select a bipartition of $m\leq n$ qubits. Doing this optimally can be exponentially expensive: not only are there $2^{m-1}-1$ possible bipartitions, but computing a bipartite entanglement or correlation metric (e.g. entanglement entropy) across a single bipartition will typically require handling a quantum state of dimension $2^m$. A scalable approach therefore depends on quantities that are cheaper to compute --- in our case, we rely upon pairwise entanglement/correlation metrics that depend only upon tractable 2-qubit states. The correct choice of metric is highly problem-dependent, and which metric is most suitable for which problem will depend highly on its properties.

In this work, we make use of two metrics: the pairwise quantum mutual information (QMI) $I(i,j)$, and the `Boolean Fourier entropy' $\kappa(i,j)$, a measure of how inseparable the dependence of the function is on bits $i$ and $j$ when only low-degree Boolean Fourier coefficients are accounted for. Both are outlined in Appendix~\ref{app:entanglement_metrics}.

In either case, for each pair of qubits we will need to access a reduced mathematical object approximately describing the 2-qubit component of the state --- for example, for the pairwise QMI, the reduced density matrix (RDM)
\begin{equation}
    \rho_{ij}=\Tr_{\overline{ij}}[\rho],\label{eqn:reduced_density_matrix}
\end{equation}
where $\overline{ij}\equiv\{1,\dots,n\}\setminus\{i,j\}$ denotes the set of all qubits except $i$ and $j$ and where $\rho=\ket{\psi}\bra{\psi}$ for $\ket{\psi}$ the state encoding the quantics function; or when computing the Boolean Fourier entropy one requires the marginalized pure state
\begin{equation}
    \ket{\mu_{ij}}=\frac{1}{\mathcal{Z}}\Big(\bigotimes_{k \in\overline{ij}}\bra{+_k}\Big)\ket{\psi},
\end{equation}
where $\mathcal{Z}$ is the normalization constant that $\ell^2$-normalizes $\ket{\mu_{ij}}$. We note that in the case where one already has an efficient TTN (including MPS) representation of the state, either of these objects can be computed efficiently without sampling error using tensor network contractions --- to compute $\rho_{ij}$, one follows the standard TTN contraction algorithm for RDMs~\cite{shi2006classical}, and to compute $\ket{\mu_{ij}}$ one contracts $\bra{+_k}$ to each open index $k\in\overline{ij}$ and then contracts the internal bonds in tree order. In cases where one does not already have an efficient representation of the state, one can approximate either using Monte Carlo sampling so long as one has efficient query access to $f(\bm{\sigma})$ (as in Section~\ref{sec:benchmarking_scent}). Therefore, viable methods exist to compute the requisite affinity matrices at least approximately in any case. We also note that the vast majority of functions of interest do not have a sign problem for Monte Carlo calculations. In many cases (e.g. financial applications, studied in Section~\ref{sec:finance_applications}) one wishes to prepare a multivariate probability distribution, which is strictly non-negative and obviously unafflicted by the sign problem. Even in cases where one wishes to prepare an antisymmetrized fermionic wavefunction in first quantization, as we outline in Section~\ref{sec:chemistry_applications} our methods could be used to prepare individual basis functions, which can then be correctly antisymmetrized using the approach of Ref.~\cite{huggins2025efficient} --- again circumventing the sign problem for Monte Carlo estimation.

When one already has a prior tensor network and wishes to compress it further to reduce state-preparation requirements, one can therefore obtain exact, noise-free affinity matrices. Without this, one can instead use Monte Carlo sampling to obtain noisy affinity matrices, apply these to compute a reasonably efficient TTN representation via \textsc{scent}, and then use this TTN to \emph{noiselessly} compute the affinity matrices, which may lead to an even more efficient TTN via a second round of \textsc{scent} --- thereby `bootstrapping' the structural optimization, and avoiding the need to construct a potentially intractable initial MPS seen in local-update-based structural optimization methods~\cite{hikihara2023automatic,manabe2025state,hikihara2025improving}. We demonstrate this later in Section~\ref{sec:benchmarking_scent}. By combining the isometric gauge with exact circuit decompositions for isometries~\cite{iten2016quantum,malvetti2021quantum,berry2025rapid}, one could already use this to produce tractable circuits to encode $f(\bm{\sigma})$ in a quantum state. Since there is already a level of error present in the function approximation with TCI (controlled via $\epsilon_{\operatorname{id}}$), it is desirable to investigate approximate circuit synthesis methods (since it would be wasteful to spend resources on compiling an approximation to greater fidelity than its underlying error). We introduce new methods for this in Section~\ref{sec:compilation}. 

\subsubsection{Rescaling affinity matrices\label{sec:rescaling_affinity_matrices}}
Binary spectral clustering approximately minimizes the mean affinity pairwise between elements of each cluster (Theorem~\ref{theorem:scent_bipartition}, Appendix~\ref{app:spectral_clustering_mean_affinity}). Since our methods support the use of a wide range of metrics for the affinity matrix --- the choice of which depends heavily on the problem structure --- this may present a challenge depending on the dynamic range and (sub-)additivity properties of the metric in question. While any monotone rescaling may work, for the purposes of this work we found that a power transform $A_{ij}\to A_{ij}^{\alpha}$ often performs well, where $A_{ij}$ are elements of the (unscaled) affinity matrix and $\alpha\in(0,1]$ is a shape hyperparameter that may need to be tuned for different metrics and problem instances.

\subsubsection{Recursive spectral clustering}
With suitably rescaled affinity matrices, we are now ready to recursively apply spectral clustering to determine the tree structure, with each bipartition of qubits determining the branching of the tree. Our approach is outlined in Algorithm~\ref{alg:scent} --- the process essentially consists of beginning with a (possibly rescaled) affinity matrix $\bm{A}$, separating its elements (the qubit indices) into two clusters with spectral clustering, then recursively performing the process on submatrices of $\bm{A}$ corresponding to the qubit indices assigned to each cluster. Each iteration of this recursive process determines a branching of the tree, which concludes when every qubit is in its own cluster --- this tree graph structure then defines the structure of the TTN.

\begin{figure}
\begin{algorithm}[H]
\caption{Spectral Clustering for Entanglement miNimizing Trees (\textsc{scent})}\label{alg:scent}
\begin{algorithmic}[1]
\Procedure{scent}{$f$}
    \ForAll{pairs $(i,j)$ of qubits}
        \State Compute affinity $A_{ij}$ (Appendix~\ref{app:entanglement_metrics})
        \State Apply monotone rescaling to $A_{ij}$
    \EndFor
    \State Construct affinity matrix $\bm{A}$ with entries $A_{ij}$
    \State \Return \Call{BuildTree}{$\bm{A}$}
\EndProcedure

\Function{BuildTree}{$\bm{A}'$}
    \If{size of $\bm{A}'$ is 1}
        \State \Return leaf node (qubit index)
    \Else
        \State Spectral clustering of indices of $\bm{A}'$ into clusters $C_1$, $C_2$
        \State $T_1 \gets$ \Call{BuildTree}{$\bm{A}'|_{C_1}$}
        \State $T_2 \gets$ \Call{BuildTree}{$\bm{A}'|_{C_2}$}
        \If{size of $\bm{A}'$ is $n$}
            \State \Return graph with edge connecting $T_1$, $T_2$
        \Else
            \State \Return node with edges connected to $T_1$, $T_2$
        \EndIf
    \EndIf
\EndFunction
\end{algorithmic}
\end{algorithm}
\end{figure}

We note that while we have presented this in terms of binary clustering, it generalizes trivially to higher-order partitioning (and our codebase supports this). We have focused on binary clustering for this work, since our numerical tests demonstrated it to have better performance.

\subsection{Numerical results\label{sec:benchmarking_scent}}

In this section, we benchmark and demonstrate essential features of our method via numerical experiments. Throughout this work, when approximating a quantics function $f(\bm{\sigma})$ by a tensor network $\mathcal{T}(\bm{\sigma})$, we consider both the size and accuracy of $\mathcal{T}$. Although the bond dimension $\chi$ is the most common metric in the literature for a tensor network's size, it is misleading for comparison between tensor networks of differing topology --- a degree-$N$ tensor with $N$ bonds will have $\mathcal{O}(\chi^N)$ amplitudes, and exact state-preparation costs for isometries scale roughly in the number of amplitudes~\cite{iten2016quantum}. Since we will be comparing tensor networks of differing structure (e.g. `comb' TTNs~\eqref{eqn:comb} have tensors up to degree $N=4$, while MPS~\eqref{eqn:mps_diagram_form} have tensors only up to $N=3$) and are focused on state-preparation costs, we will use the total number of tensor amplitudes in the network (proportional to its classical memory footprint) as a metric for TTN size throughout this work.

The error for a single index configuration $\bm{\sigma}$ may be written as $\delta_{\mathcal{T}}(\bm{\sigma})\equiv|f(\bm{\sigma})-\mathcal{T}(\bm{\sigma})|$, and we define the average error as
\begin{equation}
    \overline{\delta_{\mathcal{T}}} \equiv \mathbb{E}\left[\delta_{\mathcal{T}}(\bm{\sigma})\right],
\end{equation}
where the $\bm{\sigma}$ are drawn from the underlying probability distribution (with probability mass corresponding to the $\ell^1$ or $\ell^2$ norm of $f$, depending on the application). If TCI has converged to a good approximation $\mathcal{T}$, then one can always approximately generate these samples efficiently by using tensor network sampling algorithms~\cite{stoudenmire2010minimally,ferris2012perfect}.

To benchmark and illustrate our method, we perform numerical experiments on the $D$-variable multivariate normal distribution
\begin{equation}
    \mathcal{N}_{(\vec{\mu}, \bm{\Sigma})}(\vec{x})= \frac{1}{\sqrt{(2\pi)^D |\bm{\Sigma}|}} \exp\left(-\frac{1}{2} (\vec{x} - \vec{\mu})^T \bm{\Sigma}^{-1} (\vec{x} - \vec{\mu})\right),\label{eqn:multivariate_normal_distribution}
\end{equation}
with mean values $\vec{\mu}$ and correlation matrix $\bm{\Sigma}$. This provides a versatile testbed for our methods, since the level of correlation between variables can be directly controlled via $\bm{\Sigma}$. Following the example of Ref.~\cite{tindall2024compressing}, we find it particularly convenient to parametrize this by drawing $\bm{\Sigma}$ from the $D$-variable Lewandowski-Kurowicka-Joe (LKJ) distribution $\operatorname{LKJ}_D(\eta)$, a probability distribution over $D\times D$ positive-definite correlation matrices and whose probability density is~\cite{lewandowski2009generating}
\begin{equation}
    p_{\eta,D}(\bm{\Sigma}) = c_{\eta,D}(\bm{\Sigma})\det(\bm{\Sigma})^{\eta-1},
\end{equation}
for shape parameter $\eta > 0$ and normalization constant $c_{\eta,D}(\bm{\Sigma})$. Consequently, a larger shape parameter $\eta$ corresponds to weaker inter-variable correlations (with $\eta=1$ corresponding to a uniform probability distribution over correlation matrices).

\begin{figure}
    \centering
    \includegraphics[width=0.5\textwidth]{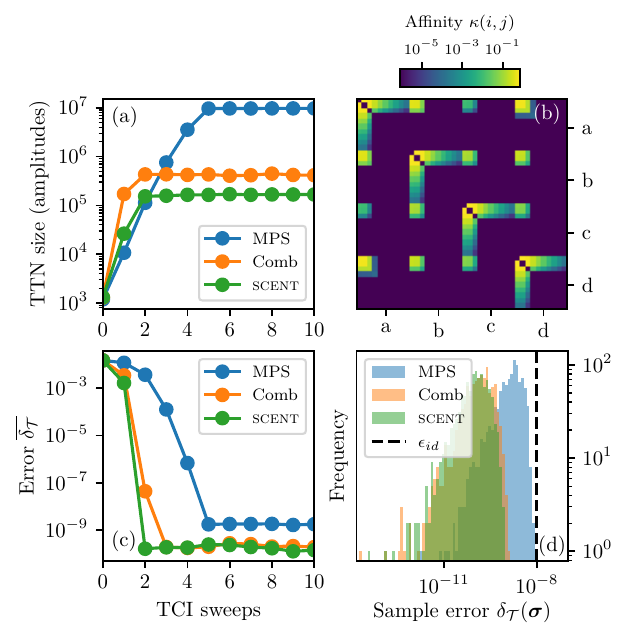}
    \caption{\textbf{Comparison of tensor network structures on an example problem.} We demonstrate the approximation of a 40-qubit quantum state encoding a strongly correlated ($\eta=1$) quadrivariate normal distribution $\mathcal{N}_{(\vec{\mu},\bm{\Sigma})}(a,b,c,d)$, for a single realization of $\bm{\Sigma}\sim \operatorname{LKJ}_{D=4}(\eta=1)$ at $B=10$ bits of precision per variable. TCI is performed with an MPS (blue), comb TTN (orange), and \textsc{scent}-optimized TTN (green) at tolerance $\epsilon_{\operatorname{id}}=10^{-8}$. \textbf{(a)} Convergence of the tensor network size in TCI --- all three topologies converge within 5 sweeps, their sizes plateauing. Relative to MPS, the comb TTN compresses the function by a factor of $23$, whereas the \textsc{scent}-optimized TTN compresses it by a factor of $58$. \textbf{(b)} Affinity matrix for \textsc{scent}, consisting of the Boolean Fourier entropy $\kappa(i,j)$ (Appendix~\ref{app:entanglement_metrics}). \textbf{(c)} Convergence of error $\overline{\delta_{\mathcal{T}}}$ in TCI --- all three topologies converge to low error within 5 sweeps, with the TTN-based approaches achieving better error. \textbf{(d)} Histogram of sample errors $\delta_{\mathcal{T}}(\bm{\sigma})$, again showing that the TTN-based approaches converge to better approximations. The TCI tolerance $\epsilon_{\operatorname{ID}}$ is denoted as a dashed black line.}
    \label{fig:multivariate_normal_example}
\end{figure}

In Figure~\ref{fig:multivariate_normal_example}, we demonstrate the approximation of a single realization of the multinormal distribution~\eqref{eqn:multivariate_normal_distribution} for $D=4$ variables each at $10$ bits of precision and strong intervariable correlations $\eta=1$, which we denote as $\mathcal{N}_{(\vec{\mu},\bm{\Sigma})}(a,b,c,d)$. The tensor network approximations $\mathcal{T}$ are produced using TCI with tolerance $\epsilon_{\operatorname{ID}}=10^{-8}$ using an interleaved MPS (blue), comb TTN (orange), and \textsc{scent}-optimized TTN (green). In Figure~\ref{fig:multivariate_normal_example}(a), we show the convergence of TCI for these three topologies over 10 sweeps, with the size of each tensor network converging within 5 sweeps. In line with the results of Ref.~\cite{tindall2024compressing}, we see that the comb TTN achieves significant compression versus MPS ($23\times$ smaller). Improving further upon these capabilities, we find that \textsc{scent} compresses the function even further ($58\times$ smaller) while also achieving a slight reduction in the approximation error (see \ref{fig:multivariate_normal_example}(d)). In Figure~\ref{fig:multivariate_normal_example}(b), we show the affinity matrix used to optimize the TTN structure with \textsc{scent} --- in this case, the Boolean Fourier entropy $\kappa(i,j)$ outlined in Appendix~\ref{app:entanglement_metrics}. Here, the largest-scale structure can be seen visually: we see that the intervariable correlation between $a$ and $c$ is weaker than any other pair of distribution variables. This leads \textsc{scent} to discover a TTN structure in which the bits of $a$ and $c$ are more distant, placing the most strongly correlated bits close together and therefore minimizing bond crossings. Even in this simple case, this enables a much more parsimonious approximation of $\mathcal{N}_{(\vec{\mu},\bm{\Sigma})}(a,b,c,d)$. In Figure~\ref{fig:multivariate_normal_example}(c) we see the convergence of error in TCI for the three network topologies, as $\overline{\delta_{\mathcal{T}}}$ rapidly decreases with sweeps until convergence is achieved within 5 sweeps. Here, $\overline{\delta_{\mathcal{T}}}$ is sampled from 1000 index configurations randomly drawn from $\mathcal{N}_{(\vec{\mu},\bm{\Sigma})}(a,b,c,d)$ using Cholesky decomposition. While all three network topologies converge, we see that \textsc{scent} and the comb TTN achieve slightly lower error than the MPS despite being much more efficient representations. We note that our approach of minimizing entanglement at each branch point is expected to results in TNs with low size, but there is no guarantee this will also reduce the error at convergence. Nonetheless, as also observed in Ref.~\cite{tindall2024compressing}, in some problems it fortunately does. In Figure~\ref{fig:multivariate_normal_example}(d), we plot histograms of the sample errors $\delta_{\mathcal{T}}(\bm{\sigma})$ of the final converged tensor networks (after 10 TCI sweeps), observing the full distribution of errors and noting how the comb and \textsc{scent} TTN structures achieve a median error approximately two orders of magnitude better than that achieved with an MPS.

In Figure~\ref{fig:size_and_correlation}, we study the multinormal distribution~\eqref{eqn:multivariate_normal_distribution} in a more systematic manner, studying the effect of the number of variables and the correlation strength on the size of the TTN. In Figure~\ref{fig:size_and_correlation}(a), where we study the effect of the variable count at fixed $\eta$ (moderate inter-variable correlation), we see that both comb and \textsc{scent}-optimized TTNs achieve significantly better scaling than MPS, with the advantage of \textsc{scent} becoming particularly apparent from $D=4$ onwards. For even a small number of variables, the advantage of TTN methods versus MPS encoding can yield a reduction of orders of magnitude in the amplitudes required, which translates directly to circuit cost reductions~\cite{iten2016quantum}. In Figure~\ref{fig:size_and_correlation}(b), where we study the effect of correlation strength at $D=4$ variables, we see again that both comb and \textsc{scent}-optimized TTNs perform substantially better than MPS, with \textsc{scent} achieving additional compression compared to the state of the art reported in Ref.~\cite{tindall2024compressing}. In all cases the TTN-based methods perform better than MPS by several orders of magnitude; furthermore, we see that TTN-based methods have considerably reduced complexity at weaker correlation strengths (higher $\eta$), whereas MPS methods achieve almost no reduction in size between the strongest and weakest correlations studied. This strongly supports that appropriately structured TTNs are much better tailored than MPS to the entanglement structure of multivariate functions encoded in quantum states.

\begin{figure}
    \centering
    \includegraphics[width=0.5\textwidth]{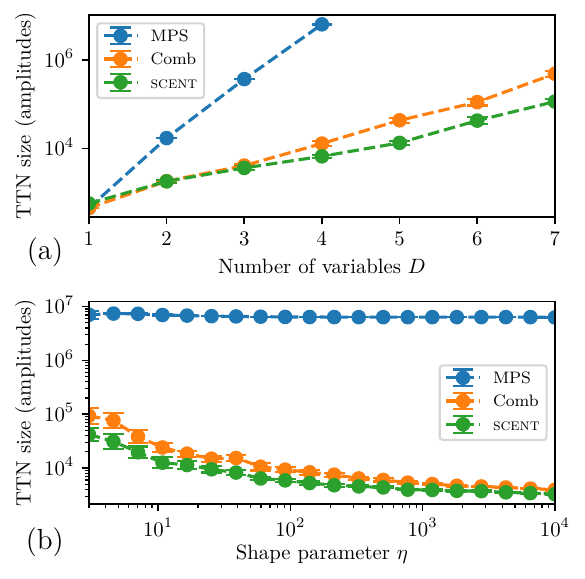}
    \caption{\textbf{Effect of variable count and correlation strength on tensor network function construction.} We study systematically the effect of the number of variables $D$ and inverse correlation strength $\eta$ on multivariate normal distributions $\mathcal{N}_{(\vec{\mu},\bm{\Sigma})}(\vec{x})$. For each data point, 20 random covariance matrices are sampled from $\bm{\Sigma}\sim \operatorname{LKJ}_{D}(\eta)$, and 10 sweeps of TCI (sufficient for convergence, see e.g. Figure~\ref{fig:multivariate_normal_example}(c)) are performed to tolerance $\epsilon_{\operatorname{id}}=10^{-8}$ with $B=10$ bits of precision per variable. Error bars represent bootstrapped 95\% confidence intervals in the logarithmic mean. \textbf{(a)} The effect of variable count $D$ on tensor network size at moderate correlation $\eta=50$. \textbf{(b)} The effect of correlation strength on tensor network size for $D=4$ variables.}
    \label{fig:size_and_correlation}
\end{figure}

As noted in Section~\ref{sec:computing_entanglement_metrics}, the affinity matrices for \textsc{scent} can be constructed efficiently and noiselessly if one already has access to an MPS/TTN representation of the function, which can be used to further compress the function if a more suitable structure is found by \textsc{scent}. However, in many cases an MPS representation of a function will be impractically large even when an efficient TTN representation exists, creating a `chicken-and-egg' problem whereby one needs to somehow first bootstrap their way to a suitable structure. As noted previously, this can be achieved by constructing the affinity matrix with Monte Carlo sampling, which unlike initializing with an MPS does not depend on the level of inter-variable correlation in the function. In Figure \ref{fig:monte_carlo_convergence}, we demonstrate the convergence of Monte Carlo sampling of the affinity matrix for \textsc{scent}. In Figure \ref{fig:monte_carlo_convergence}(a), we see the improvement in the affinity matrix with increasing sample count $N_s$, with strong agreement with the noiseless case for all but the weakest correlations by $N_s=10^4$. Strong correlation between bits of high significance can be determined accurately with relatively few samples; weaker correlations between low-significance bits require more samples to accurately resolve, but are of comparatively little importance when determining an efficient network structure. In Figure \ref{fig:monte_carlo_convergence}(b), we compare convergence traces for TCI using network structures determined with \textsc{scent} at a range of Monte Carlo sample counts $N_s\in [150,10^4]$. For this problem instance, while a minimum of $N_s\approx 10^3$ samples seems necessary for TCI to converge to a tractable network structure, by $N_s=10^4$ samples the \textsc{scent}-optimized network structure performs almost as well as in the noiseless case. In this case, we see empirically that it is feasible to bootstrap optimization by first executing \textsc{scent} on a Monte-Carlo-sampled affinity matrix to get an initial TTN, use this to obtain a noiseless affinity matrix, and then use this latter affinity matrix to perform optimization unafflicted by sampling noise. In doing so, we are able to overcome the `chicken-and-egg' problem of local-update-based structural optimization by entirely bypassing the initial MPS using few random queries of the full function. We stress that this process does not compete with structural optimization methods based on local updates~\cite{hikihara2023automatic,hikihara2025improving} --- indeed, executing these methods on the \textsc{scent}-optimized TTN may provide further refinements, a possibility we discuss in Section~\ref{sec:conclusions}.

\begin{figure}
    \centering
    \includegraphics[width=0.5\textwidth]{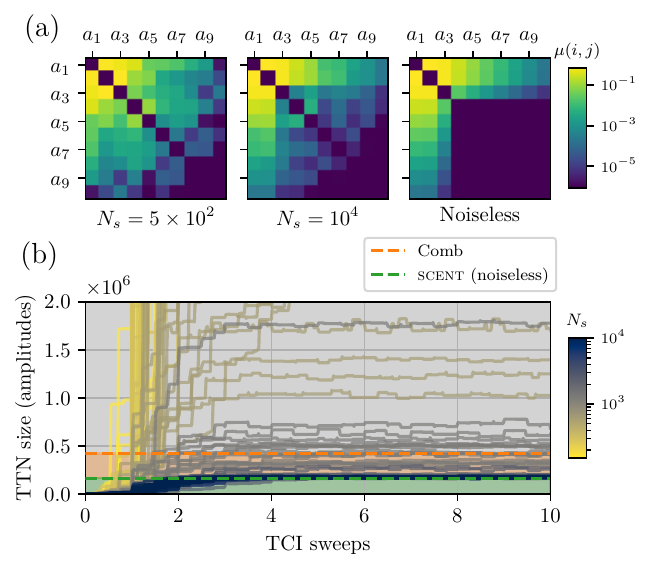}
    \caption{\textbf{Convergence of Monte Carlo sampling to construct the affinity matrix.} Here, we study the approximation of a 40-qubit quantum state encoding a strongly correlated quadrivariate normal distribution $\mathcal{N}_{(\vec{\mu},\bm{\Sigma})}(a,b,c,d)$, for a single realization of $\bm{\Sigma}\sim \operatorname{LKJ}_{D=4}(\eta=1)$ (as in Figure~\ref{fig:multivariate_normal_example}) using both Monte Carlo and exact methods for \textsc{scent}. \textbf{(a)} Single-variable submatrices of affinity matrices constructed with Monte Carlo for $N_s=5\times 10^2$ samples (left), $N_s=10^4$ samples (middle), and noiseless tensor network methods (right). \textbf{(b)} Convergence traces for TCI using network structures determined with \textsc{scent} for varying Monte Carlo sample counts $N_s$. As a point of comparison, we also plot the final TTN size for noiseless \textsc{scent} (green dashed line) and a comb network (orange dashed line).}
    \label{fig:monte_carlo_convergence}
\end{figure}

Taken together, the methods of this section enable us to significantly compress the tensor-network representations of multivariate functions. While we will introduce enhanced methods for circuit synthesis in Section~\ref{sec:compilation}, we note that this is already greatly beneficial for state preparation using standard circuit synthesis methods. Since any TTN (including an MPS) can be easily placed in isometric gauge by the methods of Section~\ref{sec:ttn_methods}, we can compile a TTN (up to normalization) to a state-preparation circuit by using standard quantum circuit decompositions for isometries \cite{knill1995approximation,iten2016quantum,malvetti2021quantum}. Reducing the number of amplitudes constraining the quantum state reduces the resources required for quantum state preparation, since the entangling gate count in these decompositions scales roughly proportionately to the number of amplitudes in the isometry (in the absence of special structure like sparsity \cite{malvetti2021quantum}).

In Figure~\ref{fig:size_vs_compilation_costs}, we demonstrate this empirically. Once again, we study the multinormal distribution~\eqref{eqn:multivariate_normal_distribution} for $D=4$ and varying $\eta$, this time quantifying the circuit costs for state preparation. For all numbers here we use the improved Knill decomposition of Ref.~\cite{iten2016quantum}, which we found to be more efficient than alternative decompositions for the tasks at hand. We sample many realizations of the distribution $\mathcal{N}_{(\vec{\mu},\bm{\Sigma})}(a,b,c,d)$ with $\bm{\Sigma}\sim \operatorname{LKJ}_{D=4}(\eta)$ at varying $\eta$, apply the \textsc{scent} protocol, and then perform TCI at tolerance $\epsilon_{\operatorname{id}}=10^{-8}$ to produce a highly-optimized TTN approximation of each realization. We then follow the procedures of Section~\ref{sec:ttn_methods} to place the TTN in isometric gauge centred on the lowest-eccentricity tensor. This represents the state as a sequence of isometries with tree-like causal structure, enabling it to be compiled to a quantum circuit by applying standard circuit decompositions to each isometry. This allows us to produce CNOT counts and circuit depths required for state preparation of the TTN representations, which we plot as filled circles in Figure \ref{fig:size_vs_compilation_costs}(a). When the same realizations of $\mathcal{N}_{(\vec{\mu},\bm{\Sigma})}(a,b,c,d)$ are instead represented by MPS in this process (hollow circles), state preparation requires a vastly higher CNOT count (between roughly $10^2$ and $2\times10^3$ times more expensive in the observed samples). The effect is more pronounced at weaker inter-variable correlations (higher $\eta)$, since the remaining entanglement between highly-significant bits is long-range in MPS but short-range in the \textsc{scent}-optimized TTN (c.f. Figure \ref{fig:size_and_correlation}(b)). The circuit costs in the MPS case are roughly constant in $\eta$, since the MPS structure is not well-suited to exploitation of the weaker correlation structure --- in comparison, the TTN structure is able to exploit this simpler structure to reduce the amplitude count. An even more pronounced effect is seen in the corresponding circuit depths in Figure~\ref{fig:size_vs_compilation_costs}(b), where the \textsc{scent}-optimized TTN is able to achieve circuit depths that are almost $10^4$ times shallower than their MPS counterpart. This is because in addition to the circuit depth savings of individual isometries, the multi-branching tree-like causal structure of the isometries in the TTN case allows for far greater parallelization than the MPS case, where the isometries can be parallelized over only two branches (in centre-canonical gauge) of sequential operations.

\begin{figure}
    \centering
    \includegraphics[width=0.5\textwidth]{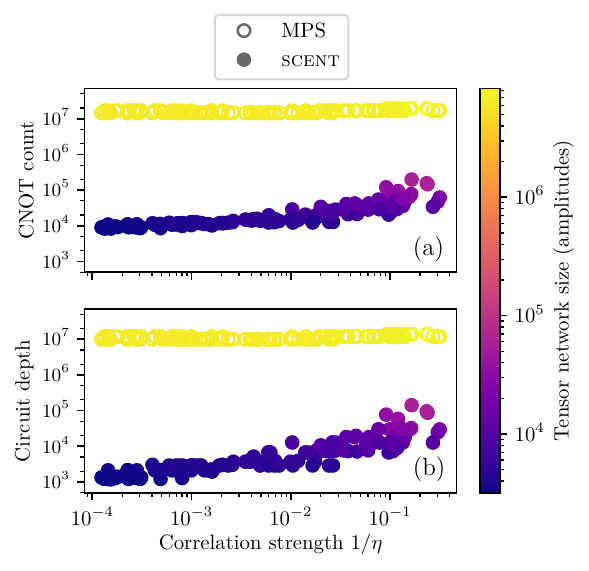}
    \caption{\label{fig:size_vs_compilation_costs}\textbf{Effect of tensor network compression on circuit synthesis costs.} Here, we study the circuit resources required with standard decompositions~\cite{iten2016quantum} to prepare a 40-qubit quantum state encoding a quadrivariate normal distribution $\mathcal{N}_{(\vec{\mu},\bm{\Sigma})}(a,b,c,d)$, for $\bm{\Sigma}\sim \operatorname{LKJ}_{D=4}(\eta)$ sampled at varying $\eta$, demonstrating that tensor network compression directly improves circuit costs. \textbf{(a)} For \textsc{scent}-optimized TTNs (filled circles), we see that CNOT count scales in the correlation strength. When the same distributions are prepared via an MPS representation (hollow circles), a drastically larger CNOT count is required (between roughly $10^2$ and $2\times10^3$ times more expensive for the observed samples), at cost nearly constant in $\eta$ --- as the MPS representation fails to take advantage of the weakened correlation strength. \textbf{(b)} A similar effect is observed for circuit depth, albeit even more pronounced (up to almost $10^4$ times shallower circuits) due to the greater parallelization afforded by branching in a tree-like causal structure.}
\end{figure}

Overall, we see that the methods of this section are greatly advantageous for state preparation of multivariate functions, improving gate counts and circuit depths by orders of magnitude compared to MPS methods even when standard isometry circuit decompositions are used. In the section that follows, we achieve further gains on end-to-end state preparation by improving the circuit-synthesis component.

\section{Gauge-aware circuit synthesis of tree tensor networks\label{sec:compilation}}
The methods of Section~\ref{sec:ttn_compression} enable one to represent multivariate functions as tensor networks more efficiently than would be possible with simple heuristically-chosen TTNs~\cite{tindall2024compressing} (and substantially more efficiently than would be possible with an MPS). Now that we have the means to prepare and compress efficient TTN representations of quantum states encoding multivariate functions, we now introduce a novel approach for synthesizing quantum circuits that prepare these states.

State preparation of matrix product states is well-studied~\cite{schon2005sequential,iten2016quantum,ran2020encoding,rudolph2022decomposition,gundlapalli2022deterministic,melnikov2023quantum,malz2024preparation,berry2025rapid}, and typically begins with the observation that an MPS in isometric gauge consists of isometries, which are physically implementable maps in qubit systems when $\dim(\mathcal{H}_{\operatorname{in}})=2^{n_{\operatorname{in}}}$ and $\dim(\mathcal{H}_{\operatorname{out}})=2^{n_{\operatorname{out}}}$ for $n_{\operatorname{in}},n_{\operatorname{out}}\in \mathbb{N}_0$ (i.e. they are powers of 2) and $n_{\operatorname{out}} > n_{\operatorname{in}}$. Physically, such a map corresponds to a unitary operation on $n_{\text{out}}$ qubits, where $n_{\text{in}}$ qubits are in an arbitrary (possibly entangled) state and the $(n_{\text{out}}-n_{\text{in}})$ remaining qubits are guaranteed to be initialized in a known state (typically the computational zero state). Combined with exact circuit decompositions for isometries~\cite{knill1995approximation,iten2016quantum,malvetti2021quantum}, or uncontrolled approximate methods such as unitary disentanglers~\cite{ran2020encoding,rudolph2022decomposition,bohun2024entanglement,sugawara2025embedding}, one can prepare an MPS with a cost scaling with its bond dimension.

Since a general TTN can be placed in isometric gauge, a straightforward generalization of this process can be used to prepare low-bond-dimension TTNs on quantum computers, as studied in Figure~\ref{fig:size_vs_compilation_costs}. While this has previously been applied to multivariate function preparation~\cite{manabe2025state}, several issues remain with this approach. First, prior approaches rely upon exact circuit decompositions for isometries~\cite{knill1995approximation,iten2016quantum,malvetti2021quantum} --- since the underlying function typically is only accurate to order $\epsilon_{\operatorname{id}}$ due to underlying approximations in the factorization step of TCI, exact circuit decompositions can be needlessly wasteful in terms of required depth and entangling gate count. It is therefore desirable to spend this flexibility granted by the existing approximation error by using approximate circuit synthesis methods that reduce circuit costs. However, standard approximate methods for MPS synthesis based on unitary disentanglers \cite{ran2020encoding,rudolph2022decomposition,bohun2024entanglement} are more challenging for TTNs~\cite{sugawara2025embedding} --- and furthermore introduce uncontrolled error, making it challenging to reduce circuit overheads within the tolerance $\epsilon_{\operatorname{id}}$. Second, the bond dimension will rarely be exactly a power of $2$, and in prior approaches this was overcome by naively padding to the nearest power of $2$~\cite{manabe2025state} --- the cost of implementing faithful action on these physically irrelevant padded states is non-negligible and would ideally be removed. Third, for near-term applications it is desirable to tailor circuit compilation directly to the specific architecture of the target quantum computer. Fourth, the inherent gauge freedom of tensor networks has not previously been included in circuit synthesis, presenting an exciting opportunity to greatly simplify circuits.

In this section, we present a novel method for TTN circuit synthesis that resolves all four of these concerns. A detailed comparison of our circuit synthesis methods to the state of the art in the literature can be found in Appendix~\ref{app:literature_circuit_synthesis}.

\subsection{Methods\label{sec:compilation_methods}}
The basic methodology of our compilation method follows the Variable Ansatz (VAns) algorithm of Ref.~\cite{bilkis2021semi}, which chooses between a set of rules (adding, simplifying, removing and optimizing gates) to approximately synthesize a circuit. Our approach combines this with tensor-network-based fidelity calculations that take into account padding of bond dimension and the relative ease of compiling isometries compared to unitaries, as well as structural optimization options that can be directly tailored to a desired architecture and incorporate the flexibility afforded by gauge freedom. In order to apply VAns to TTN synthesis, we need to define a cost function for the fidelity of an ansatz circuit with respect to a target isometry, which we fulfill in this section.

At each stage of compilation, we synthesize a circuit for a target isometry $V:\mathcal{H}_{\operatorname{in}}\to \mathcal{H}_{\operatorname{out}}$ drawn from a TTN prepared with the methods of Section~\ref{sec:ttn_compression}. Since \textsc{scent} produces binary trees, all tensors other than the orthogonality center can be written with `left' ($L)$ and `right' ($R$) output bonds
\begin{equation}
V=\begin{array}{c}
\begin{tikzpicture}
\Vertex[x=0,y=0,position=right,size=0.25,color=yellow!30,shape=semicircle,style={rotate=90}]{V}
\draw[line width = 1] (V) -- ++(+0.6,-0.35) node[right] {$\mathcal{H}_L$};
\draw[line width = 1] (V) -- ++(+0.6,+0.35) node[right] {$\mathcal{H}_R$};
\draw[line width = 1] (V) -- ++(-0.6,+0) node[left] {$\mathcal{H}_{\operatorname{in}}$};
\end{tikzpicture}
\end{array},
\end{equation}
where we have labeled the index legs of the tensor with the associated Hilbert spaces (in contrast to the conventions of Section~\ref{sec:preliminaries}). The orthogonality center itself can be written $V=\begin{array}{c}
\begin{tikzpicture}
\Vertex[x=0,y=0,position=right,size=0.2,color=yellow!30,shape=semicircle,style={rotate=90}]{V}
\draw[line width = 1] (V) -- ++(+0.5,-0.2) node[right] {\scriptsize$\mathcal{H}_L$};
\draw[line width = 1] (V) -- ++(+0.5,+0.2) node[right] {\scriptsize$\mathcal{H}_R$};
\draw[line width = 1] (V) -- ++(+0.5,+0) node[right] {\scriptsize$\mathcal{H}_C$};
\end{tikzpicture}
\end{array}$ --- it is an isometry with trivial input dimension, which is equivalent to state preparation of a $\chi_{\operatorname{out}}$-dimensional state, and is implementable straightforwardly as a special case of the methods we detail in this section.

We seek to synthesize a circuit $U=U_N\cdots U_1$ with $N$ gates that physically implements $V$ (up to exploitable gauge freedoms which we discuss below).
In the optimization procedure described later in this section, $N$ may increase or decrease as gates are added and removed.
To avoid redundantly paying the cost for faithful compilation on states originating from padding the bond dimension to a power of $2$, we compile with respect to a set of training states. First, we take a set of states
\begin{equation}
    \mathcal{S}=\{\ket{\psi_\alpha}\}_{\alpha=1,\dots,\chi_{\operatorname{in}}},
\end{equation}
where $\mathcal{S}$ spans $\mathcal{H}_{\operatorname{in}}$. We assume without loss of generality that the training states are orthonormalized. Noting that these are states on $\mathcal{H}_{\operatorname{in}}$, we now map the operation to a physical process that can be implemented in qubits. We pad the Hilbert spaces to the nearest power of 2, such that the isometry acts on $n_{\operatorname{in}}\equiv\lceil \log_2(\chi_{\operatorname{in}}) \rceil$ input qubits and $n_{\operatorname{out}}\equiv\lceil \log_2(\chi_{\operatorname{out}}) \rceil$ output qubits, and therefore isometries can be thought of as $(2^{n_{\operatorname{out}}}\times 2^{n_{\operatorname{out}}})$ unitary operations that need only act faithfully under the promise that  $n_{\operatorname{new}}=(n_{\operatorname{out}}-n_{\operatorname{in}})$ uninitialized qubits are in the computational zero state. This is a profoundly simpler operation than a general unitary, with exponentially (in $n_{\operatorname{new}}$) lower circuit costs~\cite{iten2016quantum} --- and therefore allows great simplification compared to previous TTN synthesis methods which rely upon full unitary embedding~\cite{sugawara2025embedding}.
Our training states are embedded in the padded output space $\mathbb{C}^{2^{n_{\operatorname{out}}}}$
\begin{equation}
    \widehat{\mathcal{S}}=\{\widehat{\ket{\psi_\alpha}}\}_{\alpha=1,\dots,\chi_{\operatorname{in}}},
\end{equation}
where
\begin{equation}
\widehat{\ket{\psi_\alpha}}=\widecheck{\ket{\psi_\alpha}}\otimes\ket{\bm{0}_{\operatorname{new}}},
\end{equation}
$\widecheck{\ket{\psi_\alpha}}\in \mathbb{C}^{2^{n_{\operatorname{in}}}}$ is the padding of $\ket{\psi_\alpha}$ into the padded input space $\mathbb{C}^{2^{n_{\operatorname{in}}}}$, and $\ket{\bm{0}_{\operatorname{new}}}$ is the computational zero state of the $n_{\operatorname{new}}$ uninitialized qubits. The goal of the circuit optimization will be to ensure $U\widehat{\ket{\psi_\alpha}}\approx\widehat{V}\widehat{\ket{\psi_\alpha}}$ for all $\alpha$ with some controlled approximation error, where $\widehat{V}$ is any unitary embedding of $V$ in the padded output space. This training of $U$ with training states $\widehat{\mathcal{S}}$ should ensure faithful action on relevant states without unnecessarily constraining optimization on faithful action for irrelevant states --- while our cost function will be calculated in the padded output space, this use of training states avoids unnecessary added constraints seen in other works from padding~\cite{manabe2025state} or full unitary embedding~\cite{sugawara2025embedding}.

A final powerful simplification comes from the fact that we are not seeking to compile individual isometries, but a whole TTN state made of a treelike sequence of isometries --- and therefore we are free to exploit the unitary gauge freedom at each bond as a resource. To the best of our knowledge, this gauge freedom has not previously been exploited in any state-preparation routine. Rather than compiling $V$, we are free to instead compile $(G_L \otimes G_R)V$, where $G_L\in U(\dim(\mathcal{H}_L))$ and $G_R\in U(\dim(\mathcal{H}_R))$ are arbitrary unitaries that may be freely chosen in whatever way most greatly simplifies the compilation procedure. Then, we may contract the gauge unitaries down the tree, leaving 
\begin{equation}
\begin{array}{c}
\begin{tikzpicture}
\Vertex[x=0,y=0,label=$V$,position=right,size=0.25,color=yellow!30,shape=semicircle,style={rotate=90}]{V}
\Vertex[x=0.4,y=-0.3,label=$G_L$,position=below,size=0.25,color=blue!30,shape=circle]{GL}
\Vertex[x=0.4,y=0.3,label=$G_R$,position=above,size=0.25,color=blue!30,shape=circle]{GR}
\Vertex[x=1.0,y=-0.3,label=$V_L$,position=left,size=0.25,color=yellow!30,shape=semicircle,style={rotate=90}]{VL}
\Vertex[x=1.0,y=0.3,label=$V_R$,position=right,size=0.25,color=yellow!30,shape=semicircle,style={rotate=90}]{VR}
\Edge[lw=1](V)(GL)
\Edge[lw=1](V)(GR)
\Edge[lw=1](GL)(VL)
\Edge[lw=1](GR)(VR)
\draw[line width = 1] (V) -- ++(-0.6,+0) node[left] {};
\draw[line width = 1] (VL) -- ++(+0.4,-0.2) node[right] {};
\draw[line width = 1] (VL) -- ++(+0.4,+0.2) node[right] {};
\draw[line width = 1] (VR) -- ++(+0.4,-0.2) node[right] {};
\draw[line width = 1] (VR) -- ++(+0.4,+0.2) node[right] {};
\end{tikzpicture}
\end{array}
=
\begin{array}{c}
\begin{tikzpicture}
\Vertex[x=0,y=0,label=$V$,position=right,size=0.25,color=yellow!30,shape=semicircle,style={rotate=90}]{V}
\Vertex[x=0.6,y=-0.3,label=$\tilde{V}_L$,position=left,size=0.25,color=green!30,shape=semicircle,style={rotate=90}]{VL}
\Vertex[x=0.6,y=0.3,label=$\tilde{V}_R$,position=right,size=0.25,color=green!30,shape=semicircle,style={rotate=90}]{VR}
\Edge[lw=1](V)(VL)
\Edge[lw=1](V)(VR)
\draw[line width = 1] (V) -- ++(-0.6,+0) node[left] {};
\draw[line width = 1] (VL) -- ++(+0.4,-0.2) node[right] {};
\draw[line width = 1] (VL) -- ++(+0.4,+0.2) node[right] {};
\draw[line width = 1] (VR) -- ++(+0.4,-0.2) node[right] {};
\draw[line width = 1] (VR) -- ++(+0.4,+0.2) node[right] {};
\end{tikzpicture}
\end{array},
\end{equation}
where we have contracted $\begin{array}{c}
\begin{tikzpicture}
\Vertex[x=0,y=0,position=below,size=0.2,color=blue!30,shape=circle]{GL}
\Vertex[x=0.35,y=0,position=left,size=0.2,color=yellow!30,shape=semicircle,style={rotate=90}]{VL}
\Edge[lw=1](GL)(VL)
\draw[line width = 1] (GL) -- ++(-0.35,+0) node[left] {};
\draw[line width = 1] (VL) -- ++(+0.35,-0.15) node[right] {};
\draw[line width = 1] (VL) -- ++(+0.35,+0.15) node[right] {};
\end{tikzpicture}
\end{array}
=\begin{array}{c}
\begin{tikzpicture}
\Vertex[x=0,y=0,position=left,size=0.2,color=green!30,shape=semicircle,style={rotate=90}]{VL}
\draw[line width = 1] (VL) -- ++(-0.35,+0) node[left] {};
\draw[line width = 1] (VL) -- ++(+0.35,-0.15) node[right] {};
\draw[line width = 1] (VL) -- ++(+0.35,+0.15) node[right] {};
\end{tikzpicture}
\end{array}$ on both branches, and $\tilde{V}_L$ ($\tilde{V}_R$) is still an isometry due to unitarity of $G_L$ ($G_R$). Since bonds of higher dimension will typically be located closer to the center of the tree, exploiting this gauge freedom allows us to simplify the cost of compilation by deferring certain rotations until further down the tree, where the reduced dimension leads to far simpler circuit structures. In the optimization procedure that follows, we optimize over $G_L$ and $G_R$ with (nearly) full unitary freedom as well as the circuit's gates, allowing us to transform each isometry into the local gauge that makes compilation easiest.

With these ingredients, we are able to define the cost function for optimization. The fidelity $\mathcal{F}_V(U)\in[0,1]$ of a circuit $U$ implementing the isometry $V$ (up to gauge freedom on the child bonds) is defined as
\begin{equation}
    \label{eqn:compilation_cost_func}
    \mathcal{F}_V(U)\equiv \Re\left(\frac{1}{\chi_{\operatorname{in}}} \sum_{\alpha=1}^{\chi_{\operatorname{in}}} \widehat{\bra{\psi_\alpha}}
    \widehat{V}^\dagger (\widehat{G}_L \otimes \widehat{G}_R)  U \widehat{\ket{\psi_\alpha}}\right),
\end{equation}
where $\widehat{G}_L$ ($\widehat{G}_R$) is any unitary embedding of $G_L$ ($G_R$) on the space of qubits it acts upon. To see why this is an appropriate choice, consider the following:
\begin{theorem}[Hilbert-Schmidt distance of circuit and isometry]
    Maximizing $\mathcal{F}_V(U)$~\eqref{eqn:compilation_cost_func} minimizes the Hilbert-Schmidt (Frobenius) distance between $U$ and $V$ on the subspace spanned by $\widehat{\mathcal{S}}$.\label{theorem:hilbert_schmidt_distance}
\end{theorem}
A proof is given in Appendix~\ref{app:hilbert_schmidt_distance_proof}. Intuitively, the use of the real part in the definition of the cost function is justified by the fact that all the objects over which we optimize (the gates in $U$, as well as the gauge matrices $\widehat{G}_L$ and $\widehat{G}_R$) contain arbitrary complex phases. Furthermore, within the tensor network framework, the real part is easier to optimize than the absolute value, as explained below. Since Equation~\ref{eqn:compilation_cost_func} depends only on the training states $\widehat{\ket{\psi_\alpha}}$, in which the new qubits are initialized as $\ket{\bm{0}_{\operatorname{new}}}$, the compilation process is equally free to converge to $(\widehat{G}_L\otimes\widehat{G}_R)U$ approximating \emph{any} unitary embedding of $V$. This ensures that there is no unnecessary compilation overhead from unitary embedding of isometries --- unlike Ref.~\cite{sugawara2025embedding}, we do not actually embed isometries in unitaries and this is simply a mathematical conceptual convenience. The fidelity $\mathcal{F}_V(U)$ is depicted diagrammatically in Figure \ref{fig:compilation_cost_function}(a).

Given this fidelity, we are able to apply the VAns procedure of Ref.~\cite{bilkis2021semi}. As detailed therein, a circuit can be grown via \textsc{Insertion} rules (increasing $N$ to improve expressibility), simplified via \textsc{Simplification} rules (reducing $N$ to remove redundant gates and improve depth), and individual gates can be optimized using methods which we detail below.

\begin{figure*}
    \centering
    \includegraphics[width=6in]{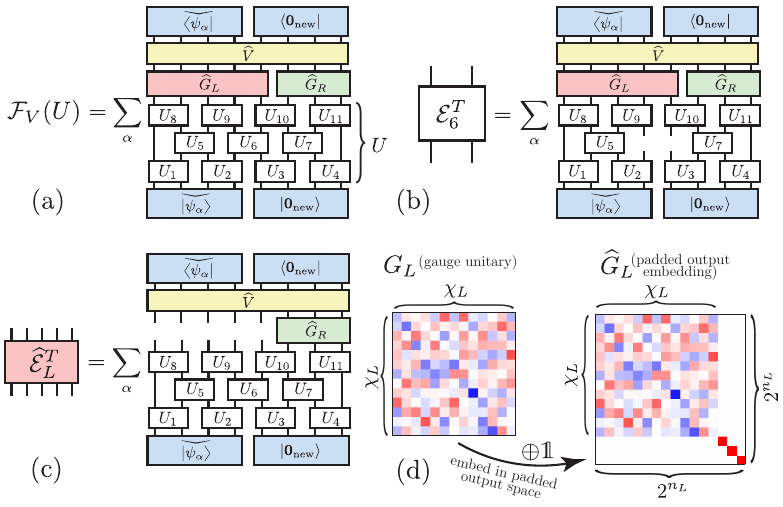}
    \caption{\textbf{Diagrammatic representation of the gate-replacement procedure.} Throughout these diagrams, note that $\sum_\alpha$ is used as a shorthand for $\frac{1}{\chi_{\operatorname{in}}}\sum_{\alpha=1}^{\chi_{\operatorname{in}}}$ for notational compactness. \textbf{(a)} Tensor-network calculation of the circuit fidelity $\mathcal{F}_V(U)$. \textbf{(b)} Tensor-network calculation of the environment tensor $\mathcal{E}$, describing the possible effect of a gate on the overall fidelity. This is obtained by removing the gate in question from the fidelity calculation, leaving its tensor indices open. \textbf{(c)} Tensor-network calculation of the (left) gauge unitary's environment tensor $\mathcal{E}_L$. As in (b), this is obtained by removing the gauge unitary and leaving its indices open. \textbf{(d)} The $\chi_L\times \chi_L$ (left) gauge unitary, whose dimension is not in general a power of $2$, is embedded in a larger $2^{n_L}\times 2^{n_L}$ unitary corresponding to actions on $n_L$ qubits.}
    \label{fig:compilation_cost_function}
\end{figure*}

Unlike related methods~\cite{bilkis2021semi,ballarin2025efficient}, we do not rely upon gradient-based optimization, but instead use tensor-network methods that are optimal for each local gate replacement (and based on a single simple matrix decomposition rather than requiring iterative convergence). Each gate is treated as an arbitrary unitary acting on some small number of qubits $n_{\operatorname{gate}}$, chosen according to the target architecture alongside iterations of the \textsc{Insertion} and \textsc{Simplification} rules. When an individual gate requires optimization, we compute an `environment tensor'~\cite{kukliansky2023qfactor,rudolph2022decomposition,gibbs2024deep}. Specifically, we take the fidelity function (Equation~\ref{eqn:compilation_cost_func} and Figure~\ref{fig:compilation_cost_function}(a)) and remove the gate tensor in question, leaving open the internal indices that were bonded to it. The resulting environment tensor $\mathcal{E}$, depicted in Figure \ref{fig:compilation_cost_function}(b), describes all possible effects of a gate at that location on the overall fidelity $\mathcal{F}_V(U)$. Finding the optimal unitary gate at that location essentially takes the form of a unitary Procrustes problem, which enables the following gate-replacement procedure:
\begin{construction}[Gate replacement]
    Let $\mathcal{E}$ be the environment tensor for an $m$-qubit gate $U_g$, computed as depicted in Figure~\ref{fig:compilation_cost_function}(b), and matricized across its bra-ket indices. Let $\mathcal{E}=XDY^\dagger$ for $X,Y\in U(2^m)$ and diagonal $D$ be a (non-economical) SVD of $\mathcal{E}$. Then, replace $U_g$ in the circuit $U$ with $YX^\dagger$.\label{construct:gate_replacement}
\end{construction}
This procedure is similar in spirit to the unitary compilation approach of Ref.~\cite{kukliansky2023qfactor}: it requires no iterative or gradient-based optimization and yields a locally optimal replacement for the gate $U_g$. Updating all gates, however, still requires iterating this procedure over the circuit, typically through multiple sweeps. The following theorem explains the local update.
\begin{theorem}[Local optimality of gate replacement]
\label{theorem:local-optimality-gate-replacement}
    The gate obtained from Construction~\ref{construct:gate_replacement} is the locally optimal gate replacement in terms of the fidelity $\mathcal{F}_V(U)$, meaning
    \begin{equation}
        U_g = \operatorname{argmax}_{U_g\in U(2^m)}\mathcal{F}_V(U).
    \end{equation}
\end{theorem}
\begin{proof}
    We first note from the tensor diagrams in Figures~\ref{fig:compilation_cost_function}(a) and \ref{fig:compilation_cost_function}(b) that $\mathcal{F}_V(U)$ can be written as
    \begin{equation}
    \mathcal{F}_V(U)=\Re(
    \begin{array}{c}
    \begin{tikzpicture}
    \Vertex[x=0,y=0,label=$\mathcal{E}$,position=left,size=0.35,color=blue!30]{A}
    \Vertex[x=0.6,y=0,label=$U_g$,position=right,size=0.35,color=blue!30]{B}
    \Edge[lw=1](A)(B)
    \Edge[bend=-135](A)(B)
    \end{tikzpicture}
    \end{array})
    =\Re(\Tr[\mathcal{E}U_g]).
    \end{equation}
    The proof then proceeds in analogy with that of Ref.~\cite{kukliansky2023qfactor}. Expanding $\mathcal{E}=XDY^\dagger$, using trace cyclicity and defining the unitary $W=Y^\dagger U_g X$, it follows that $\mathcal{F}_V(U)=\Re(\Tr[DW])$, and therefore
    \begin{equation}
        \mathcal{F}_V(U)=\sum_{j=1}^{2^m} D_{jj}\Re(W_{jj})\leq \sum_{j=1}^{2^m} D_{jj},
    \end{equation}
    since $D$ is diagonal. Since the singular values are strictly non-negative, it follows that $\mathcal{F}_V(U)$ is maximized when all $W_{jj}=1$; since $W$ is unitary this is achieved by $W=\mathds{1}$, which corresponds to $U_g=YX^\dagger$, completing the proof.
\end{proof}

This protocol can also be used to exploit the inherent gauge freedom of TTNs located at every bond.
In this approach, the (embedded) gauge unitaries $\widehat{G}_L$ and $\widehat{G}_R$ are treated as `big gates' whose positions are fixed at the end of the ansatz and which act on all qubits associated with that tensor leg (as depicted in Figure~\ref{fig:compilation_cost_function}(c)). Exploiting the high expressivity of a general many-qubit unitary in this manner allows for significant simplification at each local isometry. A small constraint is included to account for padding of local bond dimensions to powers of 2. Letting $\chi_L$ denote the left child bond, where $\chi_{\operatorname{out}}=\chi_L\chi_R$, then $n_L\equiv \lceil \log_2(\chi_{\operatorname{L}})\rceil$ qubits act as the output space for this bond, serving as the initialized qubits for the next isometry down the tree. To preserve faithful action on the padded states, the padded gauge unitary takes the form
\begin{equation}
    \widehat{G}_L = (G_L)_{\chi_{L}} \oplus (\mathds{1})_{2^{ {n_L} } - \chi_{L}}.\label{eqn:gauge_direct_sum}
\end{equation}
Therefore, when optimizing $\widehat{G}_L$, we take the $\chi_L \times \chi_L$ submatrix of the environment tensor $\mathcal{E}_L$ (corresponding to the first $\chi_L$ states in its Hilbert space), and perform the SVD-based optimization on this submatrix to obtain the optimized $G_L$. For use in further optimization sweeps, one computes the direct sum~\eqref{eqn:gauge_direct_sum} to obtain the optimized $\widehat{G}_L$. However, $G_L$ (not $\widehat{G}_L$) is contracted directly down the tree when compilation of the isometry has concluded. This constraint on the gauge unitary is depicted in Figure~\ref{fig:compilation_cost_function}(d).
The situation is analogous for the right child bond, with $n_R\equiv \lceil \log_2(\chi_{\operatorname{R}})\rceil$.

\begin{figure}
\begin{algorithm}[H]
\caption{Iterative TTN compilation}\label{alg:compiler}
\begin{algorithmic}[1]

\Procedure{CompileTTN}{$\mathcal{T}=\{T_\nu\}_\nu,k,\delta k,
                        \vec{\varepsilon}$}
    \State $\widetilde{\mathcal{T}} \gets \mathcal{T}$
    \State $\pi \gets$
        \Call{GaugeCompatibleOrder}{$\widetilde{\mathcal{T}}$}

    \ForAll{$\nu\in\pi$}
        \State \textbf{Step 1: Compile up to gauge}
        \State $(C_\nu^{(k)},G_L,G_R) \gets$
            \Call{CompileUpToGauge}{$\widetilde{T}_\nu,k,\varepsilon_1$}
        \State \Call{AbsorbGauges}
            {$\widetilde{\mathcal{T}},\nu,G_L,G_R$}
            \Comment{Push gauges to uncompiled neighbors}

        \State \textbf{Step 2: Reduce gate bodyness}
        \State $C_\nu^{(2)} \gets$
            \Call{LowerGateBodyness}{$C_\nu^{(k)},k,\delta k,\varepsilon_2$}

        \State \textbf{Step 3: Construct gate-reduction family}
        \State $\mathcal{R}_\nu \gets$
            \Call{GateReductionFamily}
                {$C_\nu^{(2)},\widetilde{T}_\nu,\varepsilon_3$}
    \EndFor

    \State $(C_\nu^\star)_\nu \gets$
        \Call{SelectCompilations}
            {$\{\mathcal{R}_\nu\}_\nu,\varepsilon_4$}
    \State \Return
        \Call{AssembleTTNCircuit}{$(C_\nu^\star)_\nu$}
\EndProcedure
\end{algorithmic}
\end{algorithm}
\end{figure}

Our TTN compilation method is presented in Algorithm~\ref{alg:compiler}. In addition to the TTN $\mathcal{T}=\{T_\nu\}_\nu$ itself, the algorithm takes several control parameters as input. The parameter $k$ denotes the size of the unitary gates used in the early stages of compilation and $\delta k$ the amount by which this gate size is reduced at each subsequent stage, while $\vec{\varepsilon}$ specifies the decomposition tolerances for the objects compiled at different stages of the algorithm. The algorithm is described in more detail in Appendix~\ref{app:TTNcompialtionAlg}.

We note that following these procedures, our circuit is compiled down to the level of arbitrary 2-qubit rotations $U_g\in \operatorname{SU}(4)$. These may be decomposed to elementary gates by use of a Cartan KAK decomposition \cite{vatan2004optimal}, leading to 3 CNOTs per rotation $U_g$. This is accounted for in all CNOT counts that follow. We also note that further reductions in the CNOT count may be possible by considering approximate two-qubit decompositions beyond standard KAK-based synthesis. With additional optimization, some two-qubit gates may be replaced by nearby unitaries that are close to the identity and admit decompositions using only two CNOT gates. We leave the exploration of such low-level compilation strategies to future work.

\subsection{Results\label{sec:compilation_results}}
With our method outlined, we now proceed to empirical demonstrations. First, in Figure~\ref{fig:single_isometry_compilation}, we demonstrate the essential building block of our procedure: the compilation of a single random isometrized tensor.

We compile a random $3 \times 8 \times 16$ isometry using the relevant components of Algorithm~\ref{alg:compiler} and demonstrate the improvement in accuracy achieved by exploiting the gauge freedom described above. Several hyperparameters govern the execution of the algorithm. These include the size $k$ of the unitary gates used during the early stage of compilation and the precision $\varepsilon$ with which these gates are decomposed into smaller gates during subsequent stages. Using a smaller $k$ reduces the number of stages but makes the optimization more prone to stalling near local minima, which can result in a suboptimal decomposition. Similarly, imposing a very small $\varepsilon$ at each stage can lead to deep decompositions. This issue is partially mitigated by the gate-removal steps performed during compilation. Nevertheless, different hyperparameter choices yield different performance, as reflected in the trade-off between decomposition error and CNOT count. We therefore typically perform several independent compilation runs with different hyperparameter settings and select the best result. This approach is illustrated in Figure~\ref{fig:single_isometry_compilation}. The plotted points represent results from several compilation runs, while the solid lines indicate the best result at each CNOT count. No single hyperparameter choice is optimal in all regimes. The best strategy depends primarily on the target decomposition error, although other considerations, such as runtime and the available degree of parallelism, may also be important.

Overall, we observe the successful approximate compilation of an example isometry at substantially lower circuit costs than standard decompositions~\cite{iten2016quantum}, where the costs can be scaled back in settings where the error tolerance is greater.

\begin{figure}
    \centering
    \includegraphics[width=0.5\textwidth]{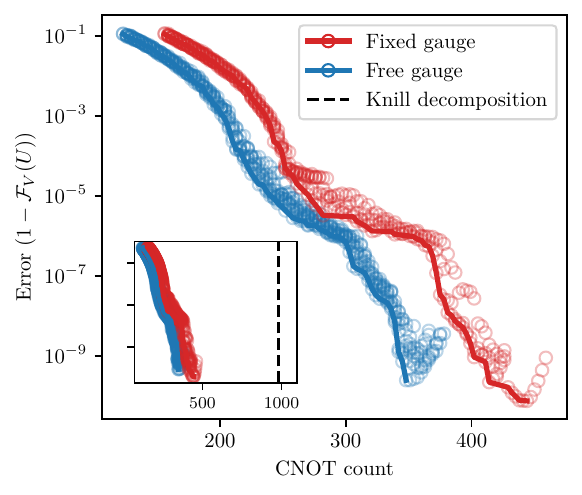}
    \caption{\label{fig:single_isometry_compilation}\textbf{Compilation of an individual isometry.} We compile a $(3\times 8 \times 16)$-dimensional isometrized tensor at a range of hyperparameter configurations, see text for details. The plot shows the error $1-\mathcal{F}_V(U)$ achieved versus the number of 2-qubit entangling gates used in the circuit decomposition. Hollow circles represent individual realizations, with the solid lines interpolating between the best realizations. We compare gauge-free compilation (blue) to fixed-gauge compilation (red), finding that the exploitation of gauge freedom consistently improves the compilation procedure. \textbf{Inset:} Both methods achieve substantially better error than standard isometry-to-circuit decompositions~\cite{iten2016quantum}, for example, $8.5\times$ fewer CNOTs at an error of $10^{-9}$ or $12\times$ fewer at an error of $10^{-5}$ in the gauge-free case.}
\end{figure}

This ability to tailor the circuit costs to the desired tolerance in fidelity, rather than waste resources on exact decompositions, is extremely powerful in our setting - since tensor network representations of amplitude-encoded functions already have some level of approximation (set for example by $\epsilon_{\operatorname{id}}$ in TCI). In Figure~\ref{fig:end_to_end_compilation}, we consider the full compilation process as applied to a function (as opposed to an individual isometry). First, each function is compressed to an efficient TTN representation using \textsc{scent} and TCI, resulting in a TTN state $\ket{\psi_{\operatorname{TTN}}}\approx e^{i\theta}\ket{\psi}$, where $\ket{\psi}$ is the target state and $\theta$ is some irrelevant global phase. This approximation error was studied at length in Section~\ref{sec:ttn_compression}; here we are concerned with how effectively we can construct a state-preparation circuit $\mathcal{U}$ for $\ket{\psi_{\operatorname{TTN}}}$, such that $|\bra{\psi_{\operatorname{TTN}}}\mathcal{U}\ket{\bm{0}}|$ is maximized. Considering 7 different realizations of the quadrivariate normal distribution $\mathcal{N}_{(\vec{\mu},\bm{\Sigma})}(a,b,c,d)$ (as studied in Figure~\ref{fig:size_vs_compilation_costs}), we run the full compilation procedure (Algorithm~\ref{alg:compiler}) for each --- this results in a compilation trace whereby infidelity decreases as the circuit is iteratively grown. For convenience, we compare these directly to the reduction in CNOT cost relative to the improved Knill decomposition of Ref.~\cite{iten2016quantum}. We observe significant reductions in CNOT cost regardless of the level of error, even when a very high-fidelity state is demanded (e.g. $\mathcal{O}(10^{-9})$ infidelity in most cases still yields a ${10}\times$ improvement in CNOT costs). In reality, there is inherent error in the TTN approximation (controlled by $\epsilon_{\operatorname{id}}$), and hardware errors (in a near-term setting) or imperfect application of non-transversal operations (in an early-fault-tolerant setting) may also result in some level of error --- it is therefore typically wasteful to spend resources on constraining infidelity beyond the already-extant envelope of error. If one is willing to tolerate some level of error, even further reductions in circuit costs can be achieved --- Figure~\ref{fig:end_to_end_compilation} can therefore be read as an empirical example of the sliding scale whereby one can trade increased error for looser error tolerance (e.g. one can achieve a CNOT cost reduction of over ${20}\times$ if infidelities below $10^{-3}$ are considered tolerable). This highlights the flexibility of our approach --- not only can we overall optimize circuit costs relative to existing methods, but we are granted further freedoms in tuning the trade-offs between cost and fidelity, rather than wasting resources on constraining exact decompositions of state representations that are themselves only approximate.

\begin{figure}
    \centering
    \includegraphics[width=0.5\textwidth]{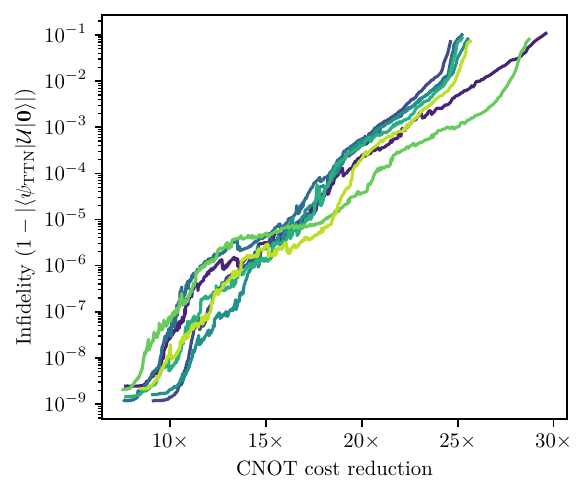}
    \caption{\label{fig:end_to_end_compilation}\textbf{Cost reductions from our compiler in end-to-end state preparation.} Here we demonstrate the successful end-to-end compilation of 7 different realizations of the quadrivariate normal distribution $\mathcal{N}_{(\vec{\mu},\bm{\Sigma})}(a,b,c,d)$ studied in Figure~\ref{fig:size_vs_compilation_costs}. For each realization, we plot the state-preparation infidelity $(1-|\bra{\psi_{\operatorname{TTN}}}\mathcal{U}\ket{\bm{0}}|)$ achieved at various levels of circuit cost reduction (in CNOT count) relative to the improved Knill decomposition of Ref.~\cite{iten2016quantum}. The cost is improved at all levels of error considered, and can be tailored to a desired level of error tolerance.}
\end{figure}

As a validation of our end-to-end methodology, in Figure~\ref{fig:treelike_correlations} we demonstrate our compilation procedure on a very challenging target: a 200-qubit state encoding a 20-variable probability distribution $\mathcal{N}_{(\vec{\mu}, \bm{\Sigma})}(\vec{x})$ with non-trivial inter-variable correlation structure. We establish the underlying probability distribution by first arranging the variables $\vec{x}=(a,\dots,t)$ randomly into a tree-like graph structure $G(V,E)$, depicted in Figure~\ref{fig:treelike_correlations}(a). The edges of this graph correspond to inter-variable correlations, and we construct the elements $\Sigma_{\alpha\beta}$ of the correlation matrix $\bm{\Sigma}$ as
\begin{equation}
    \Sigma_{\alpha\beta}=\begin{cases}
        1 & \text{if } \alpha = \beta, \\
        \rho_{\alpha\beta} & \text{if } (\alpha,\beta)\in E, \\
        0 & \text{otherwise},
    \end{cases}
\end{equation}
where for each (undirected) edge in the tree, the correlation strength is sampled independently as $\rho_{\alpha\beta} \sim \mathcal{U}_{[0.01,0.05]}$, ensuring symmetry ($\rho_{\beta\alpha} = \rho_{\alpha\beta}$). In Figure~\ref{fig:treelike_correlations}(b), we see that despite the large state, our methods are particularly well-suited to this problem due to the underlying tree-like correlation structure --- \textsc{scent} is able to efficiently re-discover the extant but unknown underlying entanglement geometry without any prior knowledge thereof. Comparing TCI on a standard comb TTN geometry~\eqref{eqn:comb} and a \textsc{scent}-optimized TTN, we find that the structural optimization from \textsc{scent} results in a ${24}\times$ compression of the TTN size relative to a comb geometry. An even larger improvement is achieved versus MPS methods (not pictured), which we were unable to quantify as the MPS representation is sufficiently large that it exceeded system memory before converging in our simulations. This shows that TTN methods are a key enabler for representing large multivariate probability distributions, and that structural optimization from \textsc{scent} can drastically improve state-preparation costs on complicated target functions. As expanded in Appendix~\ref{app:literature_circuit_synthesis}, we believe this is an improvement on the capabilities of Ref.~\cite{ballarin2025efficient}, whose large-scale demonstrations relied upon tridiagonal $\bm{\Sigma}$ and therefore nearest-neighbour 1D correlations; in contrast we retain efficient performance for much more general tree-like correlation structures. In Figure~\ref{fig:treelike_correlations}(c), we see that these improvements extend to the accuracy of the approximation: despite \textsc{scent} achieving a much more efficient TTN representation than comb geometry, it also achieves orders of magnitude lower error (average of $(6\times10^4)$ smaller sample error $\delta_{\mathcal{T}}(\bm{\sigma})$). 

In Figure~\ref{fig:treelike_correlations}(d), we show a convergence trace of the fidelity $\mathcal{F}=|\bra{\psi_{\operatorname{TTN}}}\mathcal{U}\ket{\bm{0}}|$ over the compilation process of Algorithm~\ref{alg:compiler}. As the circuit is grown, the infidelity rapidly drops, achieving a minimum infidelity of $7.44\times10^{-9}$ with just 43284 CNOTs. This is a ${1.8}\times$ improvement compared to compiling the \textsc{scent}-optimized TTN with the Knill decomposition of Ref.~\cite{iten2016quantum}; when considering the cumulative effect of the structural compression and the compiler, it is a ${43}\times$ improvement in CNOT count. It is likely that this state-preparation task would not be tractable at all with MPS methods, as we were unable to gain convergence in TCI for the same function on an MPS topology. This represents a relatively shallow ($\mathcal{O}(10^4)$ CNOTs) compilation of a 20-variable probability distribution with long-ranged, non-nearest-neighbour inter-variable correlations, at high fidelity (infidelity $\mathcal{O}(10^{-8})$) despite encoding in a large 200-qubit state. We note that relaxing fidelity requirements can yield substantially shallower circuits --- for example, one can achieve infidelities $<10^{-3}$ with as few as 5584 CNOTs. As this falls within expected capabilities of near-term hardware roadmaps, even though this is a complicated state it is reasonable to expect that these methods may enable landmark state-preparation demonstrations on near-term hardware. We emphasize that we measure the error using global fidelity, a quantity which decays exponentially with the system size. Small decomposition errors of each isometry compound rapidly and substantially reduce the global fidelity. At a level of 200 qubits, global infidelity of $10^{-3}$ becomes highly non-trivial to achieve.
Global fidelity can be a misleadingly punishing metric at larger $n$, and we expect that a global infidelity of $10^{-3}$ will be entirely sufficient for many practical applications.

\begin{figure}
    \centering
    \includegraphics[width=0.5\textwidth]{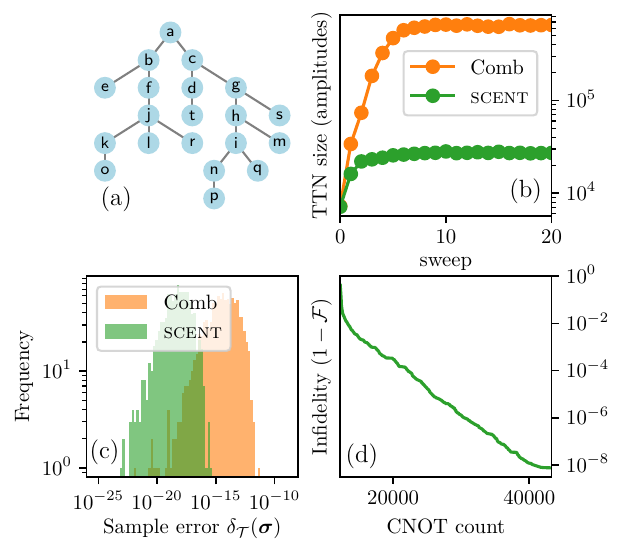}
    \caption{\label{fig:treelike_correlations}\textbf{Compilation of a 200-qubit state encoding a 20-variable correlated probability distribution.} Here, we demonstrate the use of our methods to produce a state-preparation circuit for a highly non-trivial state: a correlated multivariate normal distribution $\mathcal{N}_{(\vec{\mu}, \bm{\Sigma})}(\vec{x})$ over $D=20$ variables at $B=10$ bits of precision (200 qubits). \textbf{(a)} The variables $a,\dots,t$ are arranged into a randomly-generated tree graph $G(V,E)$, depicted here. Edges $E$ correspond to inter-variable correlations and thus non-zero values of $\bm{\Sigma}$. \textbf{(b)} TCI ($\epsilon_{\operatorname{id}}=10^{-8}$) is used to produce an efficient TTN representation of the probability distribution. The \textsc{scent}-optimized TTN (green) achieves a representation approximately ${24}\times$ more efficient than a standard comb TTN (orange); MPS is not pictured as it exceeded memory before converging. \textbf{(c)} Despite being more efficient, the \textsc{scent}-optimized TTN also achieves orders of magnitude lower error than the comb TTN, averaging $(6\times10^4)$ times smaller sample error $\delta_\mathcal{T}(\bm{\sigma})$. \textbf{(d)} Convergence trace for compilation of the \textsc{scent}-optimized TTN to a quantum state-preparation circuit. The infidelity steadily converges as the circuit grows, requiring fewer CNOTs than the decomposition of Ref.~\cite{iten2016quantum} and allowing a controllable tradeoff between circuit costs and fidelity. }
\end{figure}

We have validated our end-to-end approach to state preparation on a challenging target problem that would be impractical with existing methods. Having demonstrated the capabilities of our approach on benchmark problems, we now turn our attention to the state preparation of functions that are crucial inputs for practical end-use calculations on quantum computers.

\section{Applications}
\label{sec:applications}

Amplitude-encoded functions are an exceedingly common form of input to quantum algorithms, and state preparation of this kind is likely to be of use almost anywhere where quantum computing is used to study continuous-variable problems. Preparing such states could be extremely useful, and a non-exhaustive set of possible applications of our methods includes quantum~\cite{rebentrost2018quantum,stamatopoulos2020option} and classical~\cite{sakurai2025learning} algorithms for options pricing, quantum linear system solvers~\cite{morales2024quantum} (including high-dimensional differential equations, quantum machine learning and Green's function computation), computational fluid dynamics and nonlinear systems~\cite{krovi2023improved}, and first-quantized quantum chemistry simulations~\cite{kassal2008polynomial,su2021fault,chan2023grid,huggins2025efficient}. In this section, we demonstrate two specific examples, focusing on state preparation for chemistry simulation (Section~\ref{sec:chemistry_applications}) and preparing probability distributions for value-at-risk calculations in financial portfolio optimization (Section~\ref{sec:finance_applications}).

\subsection{Quantum chemistry\label{sec:chemistry_applications}}
The techniques described here are particularly useful for encoding initial states of grid-based first-quantized chemistry simulations.
Such simulations of chemistry have emerged as a popular application of quantum algorithms, spurred on by a seminal paper by Kassal \textit{et al.},~\cite{kassal2008polynomial} who argued that the resources for first-quantized simulations would only grow logarithmically in system size.
In comparison, second-quantized simulations which are well established in electronic structure theory~\cite{szabo1996modern} require a number of qubits that scales linearly with the system size.
A number of grid-based first-quantized quantum algorithms have been developed since, using various grid basis functions (such as plane waves~\cite{su2021fault, dutkiewicz2026spectral} and discrete variable representations~\cite{chan2023grid}) and for different applications (such as nonadiabatic dynamics~\cite{ollitrault2020nonadiabatic} and vibronic spectroscopy~\cite{feng2025quantum}).
Algorithmic improvements to first-quantized simulations have made the resource requirements and scalings more favorable~\cite{berry2025quantum, dutkiewicz2026spectral}, and some systematic comparisons with second-quantized simulations have been done~\cite{georges2025quantum}.
First quantized simulations are thus not only promising, but they may also become one of the more popular simulation frameworks.

\begin{figure*}
    \centering
    \includegraphics[width=\textwidth]{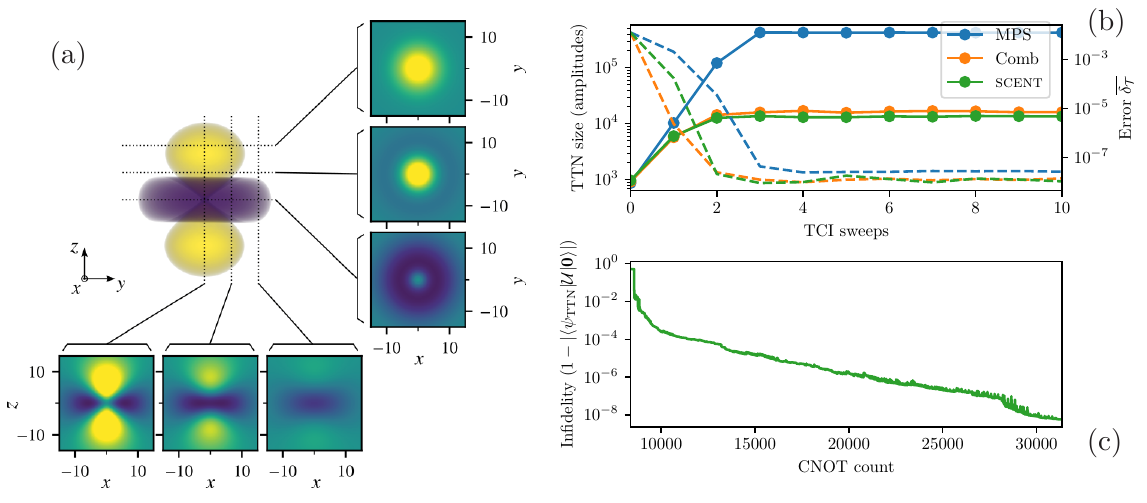}
    \caption{\textbf{State preparation of a non-Gaussian orbital for first-quantized chemistry.} Here we study the compilation of a state-preparation circuit for a hydrogenic STO. \textbf{(a)} Isosurfaces of the 3$d_{z^2}$ hydrogenic wave functions ($m_z = 0$) together with two-dimensional contour plots at different $z$ and $x$ values. \textbf{(b)} TCI ($\epsilon_{\operatorname{id}}=10^{-6}$) is used to produce tensor-network representations of the STO with MPS (blue), comb (orange) and \textsc{scent}-optimized (green) topologies. All three topologies achieve comparably small error $\overline{\delta_{\mathcal{T}}}<10^{-9}$; the TTN topologies achieve substantially more efficient representations than MPS, with \textsc{scent} proving most efficient (comb outperforms MPS by ${26}\times$, \textsc{scent} outperforms MPS by ${31}\times$). \textbf{(c)} Convergence trace of compilation of the \textsc{scent}-optimized TTN to a quantum state-preparation circuit. The infidelity steadily converges as the circuit grows, requiring fewer CNOTs than the decomposition of Ref.~\cite{iten2016quantum} and allowing a controllable tradeoff between circuit costs and fidelity.
    }
    \label{fig:chemistry_figure}
\end{figure*}

When preparing the electronic wave function of a system with more than one electron, the anti-symmetry of the wave function needs to be properly accounted for. Huggins \textit{et al.} recently showed how to efficiently prepare an approximation to the Hartree-Fock state~\cite{huggins2025efficient}. Their method uses a circuit compilation to represent each molecular orbital in a plane wave basis on a quantum register and subsequently builds the Hartree-Fock Slater determinant from the basis function registers, ensuring anti-symmetry by design.
They demonstrated their approach on a Cartesian Gaussian-type atomic orbital (GTO) basis set of the form
\begin{equation}
\label{eq:def-cartesian-gto}
    g_{lmn}^\gamma(x,y,z) \propto x^l y^m z^n e^{-\gamma r^2},
\end{equation}
where $l, m, n$ are integers, $r = \sqrt{x^2 + y^2 + z^2}$, and the exponent $\gamma$ tunes the width of the basis function. In a molecular simulation, each atom is assigned a collection of basis functions with different exponents and different values of $l$, $m$ and $n$ that are usually taken from some standard basis set~\cite{pritchard2019new}.

Cartesian GTOs are just one of many choices for atomic basis functions. They were a suitable choice for the Huggins \textit{et al.} approach since they are factorizable in Cartesian coordinates, which is a key property used in the circuit compilation to load the basis functions onto quantum registers.

However, one may wish to work with a different basis set like spherical GTOs, where the Cartesian prefactor $x^l y^m z^n$ in Eq.~\eqref{eq:def-cartesian-gto} is replaced by spherical harmonics, or Slater-type orbitals (STOs), where the Gaussian in a spherical GTO is replaced by an exponential. 
Neither of these is factorizable in Cartesian coordinates and cannot be compiled with the same strategy proposed by Huggins \textit{et al.}
With our methods, more general orbital basis functions can be loaded into a quantum register, even if they are not factorizable in Cartesian coordinates, without paying the wasteful overheads of a 3D interleaved representation (see also Appendix \ref{app:literature_comparison} for a comparison with their work). Once the basis functions are on the registers, we can use the tools of Ref.~\cite{huggins2025efficient} to build anti-symmetrized states.

In Figure~\ref{fig:chemistry_figure}, we consider the use of our methods for circuit compilation of an STO, visualized in Figure~\ref{fig:chemistry_figure}(a).
In Figure~\ref{fig:chemistry_figure}(b), we compare TCI for MPS~\eqref{eqn:mps_diagram_form}, comb~\eqref{eqn:comb}, and \textsc{scent}-optimized tensor-network geometries on this distribution. All three topologies result in similar levels of error $\overline{\delta_{\mathcal{T}}}$, however the comb topology compresses the size of the network by a factor of ${26}\times$ relative to MPS, while the \textsc{scent}-optimized topology is even more efficient at ${31}\times$ smaller than MPS, reaffirming its utility in reducing state-preparation resources. In Figure~\ref{fig:chemistry_figure}(c), we compile the \textsc{scent}-optimized TTN to a quantum state-preparation circuit with Algorithm~\ref{alg:compiler}. As the circuit grows, the infidelity decreases, eventually reaching a minimum infidelity of $5.8\times 10^{-9}$ with 31290 CNOTs, approximately 76\% of the cost of the Knill decomposition of Ref.~\cite{iten2016quantum}. As noted previously, the convergence curve of Figure~\ref{fig:chemistry_figure}(c) represents a tradeoff between circuit costs and fidelity --- e.g. one can achieve an infidelity of $10^{-4}$ with 6537 CNOTs, an improvement of ${6.3}\times$ compared to the decomposition of Ref.~\cite{iten2016quantum}. This compounds with the savings from optimizing the underlying tensor network structure in Figure~\ref{fig:chemistry_figure}(b), leading to a substantial improvement in the end-to-end state-preparation procedure.

With successful end-to-end state preparation of STOs demonstrated, one can then produce anti-symmetrized states in first quantization~\cite{huggins2025efficient}, enabling efficient initialization of first-quantized simulations on quantum computers.

\subsection{Finance\label{sec:finance_applications}}
The implementation of financial models on quantum computers is a promising application resulting from a theoretically guaranteed quadratic speed improvement over classical Monte Carlo numerical solvers~\cite{rebentrost2018quantum, orus2019quantum, stamatopoulos2020option, stamatopoulos2022towards}. Option pricing and Value-at-Risk (VaR) problems are particularly lucrative applications, as speedups in option pricing can lend a competitive edge to trading while more efficient VaR calculations can reduce computational burdens on financial institutions.

However, speedups in portfolio models rely upon two fundamental assumptions: first, that the underlying probability distribution functions of the price or return can be encoded efficiently in the register of a quantum computer, and second, that there exists an efficient oracle to implement the payoff or value function on a quantum computer. While these theoretical speed-ups are guaranteed compared to Monte Carlo methods, one should also investigate alternative classical methods that could dequantize the approach, such as tensor networks.

In this section, we consider the application of our tensor-network approaches to financial modeling to both these ends. First, and most obviously, we show the use of our methods to prepare correlated multivariate probability distributions relevant to portfolio models encoded in quantum states. Secondly, we consider (and dispel) the possibility that these same tensor-network methods might be used to dequantize the full calculation entirely, thereby further justifying the utility of a quantum algorithm in this case.

Our primary pursuit is the efficient preparation of multivariate probability distribution functions (PDF) for financial models using tree-tensor networks. Given a multivariate probability distribution $p(\vec{x})$, one wishes to prepare the state
\begin{equation}
    \ket{\psi} = \sum_{\bm{j}}\frac{\sqrt{p(\vec{x}_{\bm{j}})}}{\mathcal{Z}}\ket{\bm{j}},
\end{equation}
where we follow the notational conventions of Equation~\ref{eqn:multivariate_quantics} (i.e. $f(\vec{x}_{\bm{j}})=\sqrt{p(\vec{x}_{\bm{j}})}$). The preparation of univariate Gaussian and L\'{e}vy distributions for financial modeling using low-rank matrix product states has been demonstrated~\cite{bohun2024entanglement, iaconis2024quantum}. We consider efficient preparation of the multivariate log-normal distribution function using tree tensor networks. The multivariate log-normal distribution
\begin{equation}
    \log(\vec{x}) \sim \mathcal{N}_{(\vec{\mu}, \bm{\Sigma})}(\vec{x})\label{eqn:log_normal_distribution}
\end{equation}
is a foundational benchmark in financial modeling, since asset prices evolving under geometric Brownian motion, as in the Black–Scholes framework, are log-normally distributed~\cite{hull2018options}. Beyond equity derivatives, the log-normal law is used in portfolio modeling, Monte Carlo simulation of correlated asset dynamics, and as a baseline assumption in risk management~\cite{ mcneil2015quantitative}. 

In Figure~\ref{fig:log_normal}, we consider the end-to-end procedure for preparing these financially-relevant distributions on a quantum computer using our methods. We consider a quadrivariate ($D=4$) log-normal distribution~\eqref{eqn:log_normal_distribution} drawn from a single realization of $\bm{\Sigma}\sim \operatorname{LKJ}_{D=4}(\eta=50)$. In Figure~\ref{fig:log_normal}(a), we compare TCI for MPS~\eqref{eqn:mps_diagram_form}, comb~\eqref{eqn:comb}, and \textsc{scent}-optimized tensor-network geometries on this distribution. All three topologies result in similar levels of error $\overline{\delta_{\mathcal{T}}}$, however the comb topology compresses the size of the network by a factor of ${141}\times$ relative to MPS, while the \textsc{scent}-optimized topology is even more efficient at ${237}\times$ smaller than MPS, once again reaffirming its utility in reducing state-preparation resources. In Figure~\ref{fig:log_normal}(b), we demonstrate the compilation of the \textsc{scent}-optimized TTN to a quantum state-preparation circuit via Algorithm~\ref{alg:compiler}. We observe a steady decrease in infidelity as the circuit is grown, achieving a minimum infidelity of $8.6\times 10^{-9}$ with 5973 CNOTs. Even at this high fidelity, this represents a ${3.6}\times$ improvement over the Knill decomposition of Ref.~\cite{iten2016quantum}; as noted previously, further circuit cost savings are possible depending on the error tolerance (e.g. an improvement of ${9.9}\times$ is obtained within an infidelity tolerance of $10^{-4}$). This is in addition to the savings already achieved from the more efficient tensor-network topology achieved via \textsc{scent}. Taken together, these results validate our end-to-end ability to produce state-preparation circuits for financially-relevant optimization problems using our methodology.
\begin{figure}
    \centering
    \includegraphics[width=0.5\textwidth]{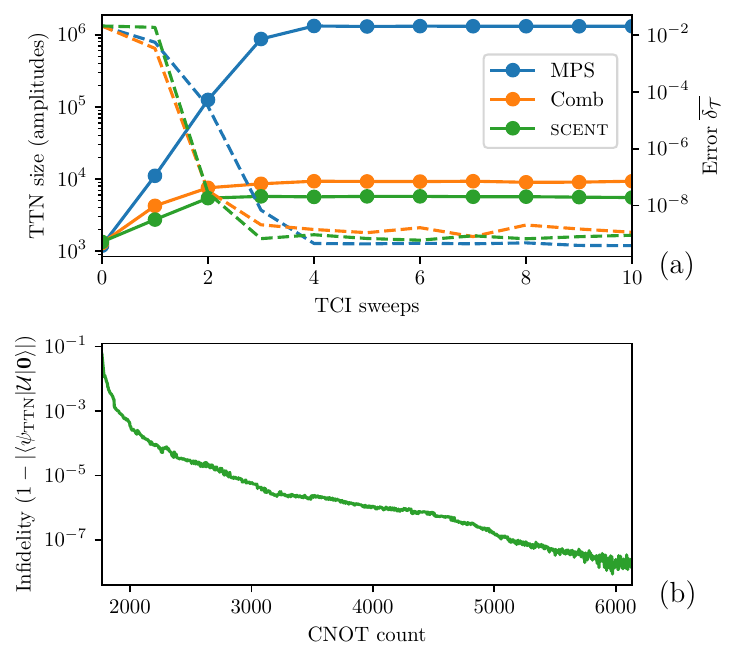}
\caption{\label{fig:log_normal}\textbf{State preparation of an asset returns distribution for financial portfolio optimization.} \textbf{(a)} TCI ($\epsilon_{\operatorname{id}}=10^{-8}$) is used to produce tensor-network representations of the asset returns distribution~\eqref{eqn:log_normal_distribution} with MPS (blue), comb (orange), and \textsc{scent}-optimized (green) topologies. All three topologies achieve comparably small error $\overline{\delta_{\mathcal{T}}}<10^{-9}$; the TTN topologies achieve substantially more efficient representations than MPS, with \textsc{scent} proving most efficient (comb outperforms MPS by ${141}\times$, \textsc{scent} outperforms MPS by ${237}\times$). \textbf{(b)}: Convergence trace of compilation of the \textsc{scent}-optimized TTN to a quantum state-preparation circuit. The infidelity steadily converges as the circuit is grown, achieving a minimum infidelity of $8.6\times10^{-9}$ with just 5973 CNOTs, a ${3.6}\times$ improvement from using the decompositions of Ref.~\cite{iten2016quantum}.}
\end{figure}

Quantics tensor networks are most commonly employed for quantum-inspired classical computations \cite{garcia2021quantuminspired,gourianov2024tensor,sakurai2025learning,niedermeier2025solving,nunez2022learning,ishida2025low,shinaoka2023multiscale,jolly2025tensorized,kim2025strong}, with efficient tensor network contractions often providing a route to dequantization. It is crucial to consider whether the same classical efficiency that enables state preparation might also dequantize our intended calculation. As an example, we now turn to the VaR problem, which requires integrating such distributions subject to discontinuous constraints. VaR calculates the depreciation of an asset collection over a time period to some desired level of confidence. Specifically, VaR defines the probability that a general loss $L(\mathbf{r})$ will exceed some threshold, $l_\alpha$ for a given confidence level $\alpha$,
\begin{equation}
\mathbb{P}\Big(L(\mathbf{r}) > l_{\alpha}\Big) = 1 - \alpha.
\end{equation}
The probability is then obtained from the cumulative distribution function
\begin{equation}
\mathbb{P}\Big(L(\mathbf{r}) > l_{\alpha}\Big) = \int_{\{\mathbf{r}: L(\mathbf{r}) > l_{\alpha}\}} f(\mathbf{r}) \, \mathrm{d}\mathbf{r},
\end{equation}
which is a constrained integral over the returns probability $f(\mathbf{r})$. We note that while we have defined VaR in return space, one can interchangeably define it in price space~\cite{stamatopoulos2020option}.
The high-dimensional nature and the addition of the loss constraint yield a difficult numerical problem which is typically solved with Monte Carlo methods. In Ref.~\cite{stamatopoulos2022towards}, the integral constraint is applied using a binary comparator function after preparing the probability distribution and the loss function on a quantum computer. Equivalently, the constrained integral can be rewritten as an unconstrained integral with an indicator function $\mathbb{I}(\bm{r})$, such that
\begin{equation}
\mathbb{P}\Big(L(\mathbf{r}) > l_{\alpha}\Big) = \int_{\mathbf{r}\in \mathbb{R}^D}\mathbb{I}(L(\mathbf{r}) > l_{\alpha})f(\mathbf{r}) \, \mathrm{d}\mathbf{r},
\label{eqn:unconstrained_integration_with_indicator}
\end{equation}
where the indicator function is defined as
\begin{equation}
    \mathbb{I}(L(\mathbf{r}) > l_{\alpha}) = \begin{cases} 
      1 &  L(\mathbf{r}) > l_{\alpha}, \\
      0 & \text{otherwise}.
   \end{cases}\label{eqn:indicator_function}
\end{equation}
While one can accommodate a small, finite number of discontinuities in the quantics representation without incurring high bond dimension~\cite{lindsey2023multiscale}, the multivariate piecewise discontinuity in Equation~\ref{eqn:indicator_function} cannot be efficiently represented. While analytically well-known, this can be demonstrated empirically by considering a `softened' indicator function
\begin{equation}
    \tilde{\mathbb{I}}(L(\mathbf{r}) > l_{\alpha}) = \frac{1}{1+e^{-\beta(L(\mathbf{r}) - l_{\alpha})}},
\end{equation}
where $\beta$ acts as an `inverse temperature' that controls function smoothness, and $\lim_{\beta\to\infty}\tilde{\mathbb{I}}(\bm{r})=\mathbb{I}(\bm{r})$. In Figure~\ref{fig:sigmoid_blowup}(a), we depict the softening of the indicator function at different values of $\beta$ graphically, for a simple three-asset portfolio with allocation weights $\bm{w}\sim \mathcal{U}_{[0,1]^3}$ and a linear loss function $L(\bm{r})=\bm{w}\cdot\bm{r}$.

\begin{figure}
    \centering
    \includegraphics[width=0.5\textwidth]{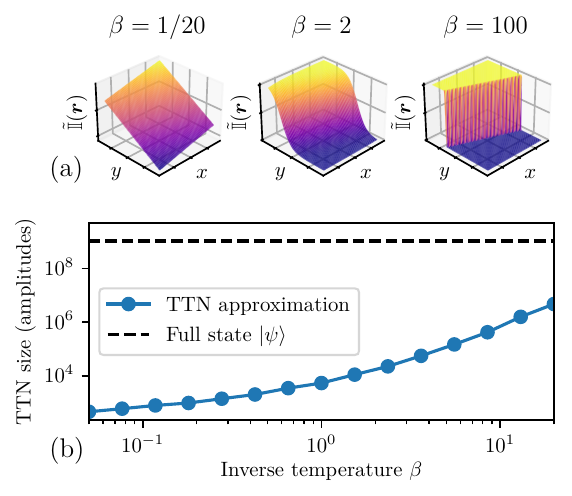}
    \caption{\textbf{Blow-up of TTN size for softened indicator functions.} We study the tensor-network approximation of a softened trivariate indicator function for a VaR problem. \textbf{(a)} Cross-sections of $\tilde{\mathbb{I}}(\bm{r})$ for $(x,y)\in [-4,4]^2$ and $z=-1$ at three different values of $\beta$. As $\beta$ increases, the approximation of the discontinuity improves.
    \textbf{(b)} The size of TTN produced by converged TCI ($\epsilon_{\operatorname{id}}=10^{-8}$, comb topology) at a range of values of $\beta$. The black dotted line represents the $2^n$ amplitudes required to store a full state vector. As $\beta$ increases, the bond dimension increases sharply, rapidly approaching the full-state cost.}
    \label{fig:sigmoid_blowup}
\end{figure}

In Figure~\ref{fig:sigmoid_blowup}(b), we study this instance of the indicator function empirically. At varying values of $\beta$, we approximate the softened indicator function $\tilde{\mathbb{I}}(\bm{r})$ with a TTN using TCI (error tolerance $\epsilon_{\operatorname{id}}=10^{-8}$). As $\beta$ increases, there is a sharp increase in the size of the TTN (in number of stored amplitudes), rapidly approaching the $2^n$ amplitudes required for a full state vector in Hilbert space. As one would anticipate analytically, approximating the indicator function for sufficiently large $\beta$ to well approximate the discontinuity in the constrained integration can be prohibitively difficult, approaching full-state costs (and therefore intractable for large $n$).

This supports the use of a quantum approach to computing VaR. Given a probability distribution that can be represented as an efficient TTN (and therefore compiled to a circuit using our methods), if the indicator function could also be represented efficiently with the same TTN structure, then one could compute the integral in Equation~\ref{eqn:unconstrained_integration_with_indicator} efficiently by contracting the tensor networks together along their shared physical indices, since an efficient contraction order exists~\cite{shi2006classical}. Crucially, in this case, the efficient tensor network methods enable state preparation but do not trivialize the full calculation --- suggesting that this may indeed be an avenue to quantum advantage on a useful financial calculation.

\section{Conclusions\label{sec:conclusions}}
In this work, we have introduced a broad set of tools for quantum state preparation with TTNs. Building upon a wealth of previous works in quantics tensor networks, we demonstrated a novel method (\textsc{scent}) for efficiently constructing optimized TTN representations of multivariate functions, with entanglement-geometry-informed structure that outperforms the previous state of the art~\cite{tindall2024compressing}. Using these efficient TTN function representations, we turned our attention towards the construction of approximate circuits to prepare these states on quantum computers. By intelligently constructing adaptive circuits for each isometry in a TTN, we demonstrated the ability to efficiently synthesize quantum circuits to amplitude-encode multivariate functions on quantum computers. Our tensor-network-based methods provide locally optimal choices of gates without the need for gradient-based or iterative updates at the gate level, avoid the additional overheads required for full unitary embedding, and allow for relaxation in the requirements of circuit synthesis within a desired error tolerance (e.g. within the margin of error set by the underlying tensor-network function approximation).

We note that our work is one part of a rich tapestry of tensor-network methods for function encoding, and further improvements can likely be made by combining our approaches with those in related papers. For example, our \textsc{scent} approach uses a single input (the affinity matrix) to construct a \emph{global approximation} to the optimal TTN structure for compressing a function --- in contrast, the `automatic structural optimization' (ASO) approach of Refs.~\cite{hikihara2023automatic,hikihara2025improving,manabe2025state} finds the \emph{locally optimal} solution to this problem in rearranging the 4 neighbouring indices of any bond in a binary TTN, which when iterated may lead to a good minimum. These approaches can be complementary rather than competing: for example, we propose that initializing a TTN with \textsc{scent} should provide an excellent initial guess that ASO can further refine, also overcoming the `chicken-and-egg' problem of needing an initial MPS in ASO (which as we have shown will be intractable for many nontrivial problems). In a similar, but opposite vein, any other TTN compilation method (e.g. exact circuits for isometries~\cite{knill1995approximation,iten2016quantum,malvetti2021quantum} or unitary disentanglers~\cite{sugawara2025embedding}) can be used to create an initial circuit that can then be refined with our methods in Section~\ref{sec:compilation}, improving their fidelity or reducing their circuit complexity within a specified error tolerance as required. The best performance is likely to arise by combining different methods, of which our work is only one.

On a broader level, we note that our methods cover all loop-free tensor network geometries, and correspondingly, we are most efficiently able to represent functions whose underlying correlation structure is tree-like (see Figure~\ref{fig:treelike_correlations}).
Moving from MPS to heuristically-chosen TTNs to optimized TTNs for function encoding has progressively improved our prospects for state preparation of multivariate functions --- following this trend, we hope that expanding these methods to more complicated looped geometries like PEPS or MERA can further expand the retinue of functions that can be prepared on quantum computers. Doing so will present both great opportunity and non-trivial technical challenges, since the interpolative gauge is not well-defined for a looped tensor network.

\section*{Acknowledgements}
We acknowledge helpful technical discussions with Gregory Boyd, B\'{a}lint Koczor, Anh Pham, and Andrew Vlasic. The code for \textsc{scent} and TCI can be found in the GitLab repository of Ref.~\cite{gitlab_repo} and was developed by the first author. We thank Vasco Ferreira for contributing to preparation of the codebase for public release. L.C. was supported by the U.S. Department of Energy, Office of Science, Office of Advanced Scientific Computing Research under Contract No. DE-AC05-00OR22725 through the Accelerated Research in Quantum Computing
Program MACH-Q. L.C. was also supported by the Laboratory Directed Research and Development program of Los Alamos National Laboratory under project number 20260043DR.

\clearpage
\newpage
\widetext
\appendix

\section*{Appendices for ``Multivariate quantum state preparation with optimized tensor networks''}
Here we present additional details for the main results in our manuscript, structured as follows.
In Appendix~\ref{app:literature_comparison}, we present a detailed comparison of our methods to the existing state of the art in the literature, with regard to both structural optimization and compression methods for loop-free tensor networks (Appendix~\ref{app:literature_compression_optimization}) and circuit synthesis for loop-free tensor network states (Appendix~\ref{app:literature_circuit_synthesis}).
In Appendix~\ref{app:entanglement_metrics}, we outline and compare different methods to construct affinity matrices for \textsc{scent}.
In Appendix~\ref{app:proofs}, we present supplementary proofs omitted from the main text for brevity.
In Appendix~\ref{app:TTNcompialtionAlg}, we present full detail on the compilation of TTNs to quantum circuits, including pseudocode.

\section{Comparison with prior literature\label{app:literature_comparison}}

\subsection{MPS/TTN structural optimization and compression methods\label{app:literature_compression_optimization}}
In most contexts where tensor networks are employed, it is desirable to optimize their structure to most efficiently represent the target object. For example, even within the linear chain topology of MPS methods, previous work has used the pairwise QMI to optimize site ordering in DMRG~\cite{rissler2006measuring,barcza2011quantum,ali2021ordering}.

For tensor network topologies beyond MPS, some methods are limited by the requirement to already have a tractable representation of the state to perform computations on. The procedure outlined in Ref.~\cite{hyatt2017extracting} discovers the entanglement geometry of a state by constructing and analyzing a unitary disentangling circuit, but requires one to perform operations on the full state. Likewise, a recent work by Okunishi \emph{et al.}~\cite{okunishi2023entanglement} considers a bipartition-based construction of TTN structure that bears similarities to our own. However, this has two key constraints that make it challenging in comparison. First, the procedure is initialized with the exponentially-sized full quantum state, which becomes impractical to store and perform matrix factorizations on beyond a small number of qubits --- entirely preventing scaling to multivariate functions of more than a few variables or of very high precision. In contrast, our methods use rank-revealing constructions that need never access the full exponentially-sized state. Second, while there are exponentially many possible bipartitions at each branching ($2^{n-1}-1$ for $n$ qubits), Ref.~\cite{okunishi2023entanglement} considers only the $n-1$ possible bipartitions that preserve a presupposed ordering. This exponential reduction in the number of available bipartitions prevents meaningful structural optimization in the case of quantics functions: if one begins with an interleaved ordering, the procedure would fail to separate variables and produce an inefficient structure that places bits of the same variable far apart in separate branches, whereas if one begins with a serial representation then the procedure would recover something akin to a comb TTN (Equation~\ref{eqn:comb}) with branches for each variable ordered in the presupposed ordering of the serial representation, therefore failing to improve on the previous state of the art~\cite{tindall2024compressing}. In contrast, our approach searches for an approximately optimal bipartition amongst all $2^{n-1}-1$ possible bipartitions.

It is therefore highly desirable to investigate techniques that do not require handling of the full state, since only these approaches can be scalable. One such technique is the generalization of TCI to TTNs by Tindall \emph{et al.}~\cite{tindall2024compressing}, bringing the utility of rank-revealing interpolative methods to the greater flexibility in entanglement structures allowed by TTNs. We use the approach introduced by Tindall \emph{et al.} extensively throughout this work. Ref.~\cite{tindall2024compressing} only considers heuristically-chosen network structures such as the comb TTN (Equation~\ref{eqn:comb}), and notes that there is likely scope to optimize the structure beyond this. We extend beyond the scope of Ref.~\cite{tindall2024compressing} by realizing this proposal with the introduction of \textsc{scent} and by employing their work as one tool in a full state-preparation pipeline for multivariate functions.

Another technique that does not require handling of the full quantum state is automatic structural optimization (ASO), introduced by Hikihara \emph{et al.}~\cite{hikihara2023automatic,hikihara2025improving}. This method `unfolds' an initial MPS into an optimized TTN structure by local optimization about each bond --- comparing the possible reorientations of the 4 indices neighbouring a bond in a binary TTN, and choosing the arrangement that minimizes entanglement. This iterative step is locally optimal, and sweeping back and forth along a tensor network structure often leads to an efficient TTN structure. Although starting from an initial MPS does prevent any need to handle the exponentially-sized full state, for multivariate functions it is often the case that an MPS representation is intractable even when an efficient TTN representation exists (see Figure~\ref{fig:size_and_correlation}(a)). In contrast, \textsc{scent} can be initialized with just an affinity matrix (which can be estimated from Monte Carlo sampling using only arithmetic query access to the function), and jumps directly to the target structure without any intermediate iterative process that may require less efficient state representations. As noted in Section~\ref{sec:conclusions}, we anticipate that a fusion of these techniques will prove useful: initializing structural optimization with \textsc{scent} and then further refining the structure with ASO could enable efficient representations of functions that are intractable with ASO alone due to the complexity of the MPS representation, as well as help initialize the iterative, locally-optimal procedure in ASO closer to a good minimum and assist its convergence.

\subsection{MPS/TTN circuit synthesis methods\label{app:literature_circuit_synthesis}}
MPS have long been understood as a promising tool for state preparation when the underlying state has suitable entanglement structure, leading to a broad range of tools for synthesizing circuits to exactly or approximately compile them~\cite{schon2005sequential,iten2016quantum,ran2020encoding,rudolph2022decomposition,gundlapalli2022deterministic,melnikov2023quantum,malz2024preparation,berry2025rapid}. When combined with the quantics representation, these have shown great promise for amplitude-encoding univariate functions on quantum computers~\cite{bohun2024entanglement}. Our work expands upon this by leveraging the more flexible structure of TTNs to more efficiently compile multivariate functions.

One particularly interesting MPS-based state-preparation method is the recent work by Huggins \emph{et al.} on state preparation for first-quantized chemistry~\cite{huggins2025efficient}. This work outlines a method to encode GTOs in Cartesian coordinates as a product of MPS, where the sites correspond to a plane-wave basis. Subsequently, they synthesize a unitary from a product of reflection operators following methods from Refs.~\cite{low2024trading, kliuchnikov2013synthesis}.
Their MPS representation relies on the fact that the Cartesian GTOs are factorizable into Cartesian coordinates. We demonstrate in Section~\ref{sec:chemistry_applications} that we can find efficient representations of non-factorizable orbitals using our methods. While we focus on grid-based representations, and Huggins \textit{et al.} focus on plane-wave representations, we expect that their techniques can be adopted to grid-based methods, allowing our respective approaches to be combined to enable first-quantized state preparation with more general orbitals.

Any TTN or MPS can be placed in isometric gauge, and therefore their preparation circuits can in principle be synthesized by using known quantum circuit decompositions for isometries~\cite{knill1995approximation,iten2016quantum,malvetti2021quantum}. Indeed, Ref.~\cite{manabe2025state} combines ASO~\cite{hikihara2023automatic,hikihara2025improving} with the circuit decompositions of Ref.~\cite{iten2016quantum} to produce state-preparation circuits for multivariate functions. First, we note that this approach inherits the comparative limitations of ASO discussed in Appendix~\ref{app:literature_compression_optimization}. Secondly, we note that the use of \emph{exact} circuit decompositions is wasteful in practice: since tensor-network function encodings involve some level of controllable approximation (e.g. for TCI as presented in this work, controlled by the ID tolerance $\epsilon_{\operatorname{id}}$), fidelity beyond a certain threshold is bottlenecked by the underlying tensor-network approximation rather than circuit synthesis --- so it is prudent to exploit this by reducing circuit costs by allowing for approximate compilation, as we do in Section~\ref{sec:compilation}. Furthermore, our approach can be directly tailored to a target architecture by varying the gate dictionary, as outlined in Ref.~\cite{bilkis2021semi}, and also benefits from exploiting gauge freedom.

Another approach to TTN circuit synthesis is that of Ref.~\cite{sugawara2025embedding}, which generalizes previous unitary disentangler approaches for MPS~\cite{ran2020encoding,rudolph2022decomposition,bohun2024entanglement} to TTNs. This approach iteratively generates a circuit by truncating the bond dimension of the TTN to $\chi=2$, generating a circuit layer from the analytically-known circuit that disentangles this truncated TTN to a product state, applying this to the TTN, and iteratively repeating --- heuristically, each layer should approximately disentangle the full TTN further, eventually converging to a good approximation of the preparation circuit's inverse. However, doing so requires one to fully embed isometries in unitaries: as discussed in Section~\ref{sec:compilation_methods}, for a (padded) isometry mapping $n_{\operatorname{in}}$ qubits to $n_{\operatorname{out}}=n_{\operatorname{in}}+n_{\operatorname{new}}$ qubits, unitary embedding incurs an additional cost that is exponential in $n_{\operatorname{new}}$ when considering exact circuit decomposition costs~\cite{iten2016quantum}. This is wasteful in practice, as many additional resources are spent on ensuring faithful action on every possible state of the $n_{\operatorname{new}}$ additional qubits despite the fact that they are guaranteed to be supplied in the computational zero state $\ket{\bm{0}_{\operatorname{new}}}$. In contrast, our objective function for optimization (Equation~\ref{eqn:compilation_cost_func}) is constructed in a manner that treats all unitary embeddings equally, imposing no particular constraint to maintain any particular one, thereby avoiding this issue. As we note in Section~\ref{sec:conclusions}, the circuit from another TTN synthesis method can in principle be used as an initial guess for our approach, which can then refine or simplify the circuit. Since Ref.~\cite{sugawara2025embedding} follows up their unitary disentangler methods with environment-tensor-based gate replacements (analogous to the MPS approach of Ref.~\cite{rudolph2022decomposition}), our methods may be a useful replacement for this step --- allowing the rough initial guess from unitary disentangling to be refined with adaptive circuit structure, platform-targeted gate choices, full awareness of gauge freedom, and a purpose-built objective function that avoids unnecessary constraints.

A recent work by Ballarin \emph{et al.}~\cite{ballarin2025efficient}, which we became aware of during preparation of this manuscript, shares several similarities with our work: their work utilizes TTNs and TCI to create an efficient representation of a multivariate function, then introduces novel circuit synthesis techniques to prepare these on quantum computers. While similar in these broad strokes, their work differs substantially on a technical level. Ref.~\cite{ballarin2025efficient} introduces interpolative quantum state preparation (IQSP), which iteratively deforms an easy-to-prepare initial function to a complicated target function, performing gradient-based optimization of a parametrized quantum circuit at each stage (and ensuring substantial gradients by warm-starting). In this paper, we do not use gradient-based optimization, instead using adaptive circuit structure and locally-optimal, environment-tensor-based gate replacement for our optimization procedure. Our circuit synthesis methods are entirely different in spirit and approach.
Another difference in scope between our work and Ref.~\cite{ballarin2025efficient} is the TTN structure employed. Ballarin \emph{et al.} consider the use of a comb TTN (Equation~\ref{eqn:comb}) only, which is suitable for variables whose correlation structure is approximately linear (i.e. strong correlations only between variables placed close on the presupposed ordering of the comb). Indeed, their largest demonstrations are for multivariate normal distributions (Equation~\ref{eqn:multivariate_normal_distribution}) with tridiagonal $\bm{\Sigma}$, such that the only non-zero correlation is between nearest neighbours in this order, a structure that is by construction particularly well-suited to the chosen comb topology. We introduce \textsc{scent} to optimize the TTN structure for an extant but undiscovered entanglement geometry, allowing functions with approximately tree-like correlations to be efficiently represented, broadening the scope of applicable functions. Indeed, this is exemplified in Figure~\ref{fig:treelike_correlations}, where we study a normal distribution with tree-like correlation patterns in $\bm{\Sigma}$ --- there, the \textsc{scent}-optimized TTN is approximately $24\times$ smaller in amplitude count compared to a comb TTN, demonstrating the significant utility of structural optimization to reduce state-preparation costs. We also note some comparative advantages of the approach by Ballarin \emph{et al.}: their interpolative approach is an inspired and useful way to initialize optimization in a favorable regime, they explicitly consider noise-aware optimization, and they demonstrate state preparation on real trapped-ion hardware.

\section{Constructing affinity matrices}
\label{app:entanglement_metrics}
The \textsc{scent} protocol introduced in Section~\ref{sec:ttn_compression} can in principle be executed with a broad range of metrics for the affinity matrix used in spectral clustering. The basic principle is to ensure that highly entangled qubits are less distant in the tree structure, and different measurements of entanglement (e.g. any entanglement monotone) or other measures of correlation can prove more suitable for different problems. Here, we define and compare several options. A good candidate will be non-negative, symmetric, take larger values when stronger entanglement or correlation is present, and have good dynamic range.

An obvious candidate for these purposes is the pairwise QMI, given by
\begin{equation}
    I(i,j)=S(\Tr_i[\rho_{ij}])+S(\Tr_j[\rho_{ij}])-S(\rho_{ij}),\label{eqn:pairwise_qmi}
\end{equation}
where $\rho_{ij}$ is the reduced density matrix at sites $i$ and $j$ \eqref{eqn:reduced_density_matrix} and $S(\rho)\equiv-\Tr[\rho \log \rho]$ is the von Neumann entropy. This has precedent in structural optimization of tensor networks as a means to optimize site ordering for DMRG calculations \cite{rissler2006measuring,barcza2011quantum,ali2021ordering}, as well as to optimize ansatz connectivity for VQE \cite{tkachenko2020correlation}.

Given access to $\rho_{ij}$, another tempting choice of metric is the entanglement of formation~\cite{bennett1996mixed}
\begin{equation}
    E_f(\rho_{ij})=\inf\left(\sum_\alpha p_\alpha E_f(\ket{\psi_\alpha}_{ij})\right),\label{eqn:entanglement_formation_defn}
\end{equation}
where the infimum is taken over all possible pure-state decompositions of the state $\rho_{ij}=\sum_\alpha p_\alpha \ket{\psi_\alpha}_{ij}\bra{\psi_\alpha}_{ij}$ and for a pure state $E_f(\ket{\psi}_{ij})=S(\Tr_i[\rho_{ij}])=S(\Tr_j[\rho_{ij}])$. For a 2-qubit state, this can be computed straightforwardly in terms of the concurrence
\begin{equation}
    \mathcal{C}(\rho)\equiv \max(0, \lambda_1 - \lambda_2 - \lambda_3 - \lambda_4),
\end{equation}
where $\lambda_1,\dots,\lambda_4$ are eigenvalues (in decreasing order) of the Hermitian operator
\begin{equation}
    R=\sqrt{\sqrt{\rho}\tilde{\rho}\sqrt{\rho}},
\end{equation}
with spin-flipped state of $\rho$ defined $\tilde{\rho} \equiv (\sigma_y \otimes \sigma_y)\rho^*(\sigma_y \otimes \sigma_y)$ and complex conjugation $\rho^*$ taken in the computational basis ($\sigma_z$ eigenbasis). The entanglement of formation is then given by
\begin{equation}
    E_f=h\left(\frac{1+\sqrt{1-\mathcal{C}^2}}{2}\right),
\end{equation}
for Shannon entropy function $h(x)=-x\log_2(x)-(1-x)\log_2(1-x)$. We trialled this for several problems and found it to be a poor choice for \textsc{scent} in practice, since the RDMs for pairs of all but the most significant bits are typically highly mixed, meaning that these localized states could be prepared without entanglement and therefore leading to $E_f=0$.

For amplitude-encoded, real-valued functions in the quantics representation, where qubits correspond to bits of binary significance in discretized continuous variables, a particularly useful measure of pairwise correlation between qubits can be devised using Boolean function theory. First, note that a real-valued $n$-qubit amplitude-encoded quantics function may be thought of as a Boolean function $f:\{0,1\}^{n} \to \mathbb{R}$. Boolean Fourier analysis of such functions is a well-established field~\cite{o2021analysis}, and we note that any such function may be written in the form
\begin{equation}
    f(x)=\sum_{S \subseteq [n]} \hat{f}(S)\chi_S(x),
\end{equation}
where $x\in \{0,1\}^n$, the $\chi_S(x)=(-1)^{\sum_{i\in S}x_i}$ are parity functions (i.e. they are characters on the group $(\mathbb{Z}_2)^n$), and the $\hat{f}(S)$ are Boolean Fourier coefficients that define the function. Given two qubits labelled by indices $i$ and $j$, we may gain first-order information about the correlation between these qubits in $f$ by considering only terms with $S\subseteq \{i,j\}$. Marginalizing the function with respect to all other qubits $[n]\setminus \{i,j\}$ isolates these Fourier coefficients, yielding
\begin{equation}
    \mu(x_i,x_j)=\frac{1}{2^{n-2}}\sum_{x_k}^{k\in \overline{ij}}f(x_k)=\sum_{S\subseteq\{i,j\}}\hat{f}(S)\chi_S(x_i,x_j),
\end{equation}
where the second equality follows from the fact that $\sum_{x_k=0}^1 \chi_S(x_k)=0$ vanishes for $k\in S$, since the parity change from a bit-flip ensures one term is $-1$ and the other term is $+1$. Put simply, averaging over a qubit cancels out the Fourier components with support on that qubit. Writing $\mu(x_i,x_j)$ as a (potentially unnormalized) 2-qubit state $\ket{\mu_{ij}}$, a suitably normalized metric of the pairwise correlation between qubits $i$ and $j$ can be obtained from its entanglement entropy
\begin{equation}
\kappa(i,j)=S(\Tr_i[\ket{\mu_{ij}}\bra{\mu_{ij}}])=S(\Tr_j[\ket{\mu_{ij}}\bra{\mu_{ij}}]).\label{eqn:boolean_fourier_entropy}
\end{equation}
This quantifies how inseparable the dependence of $f$ is on those two bits after marginalization, i.e. how inseparable the remaining Fourier support is, and we therefore refer to it as the `Boolean Fourier entropy'. Since smooth functions lead to reduced entanglement between distant bits in the quantics representation~\cite{holmes2020efficient,ali2023approximation,ali2024multivariate,bohun2024entanglement,holmes2020entanglement,garcia2021quantuminspired,marin2021quantum,lindsey2023multiscale}, it is reasonable to expect that the Fourier coefficients with support on many qubits will be much smaller and less relevant to determining entanglement geometry.

\begin{figure*}
    \centering
    \includegraphics[width=0.7\textwidth]{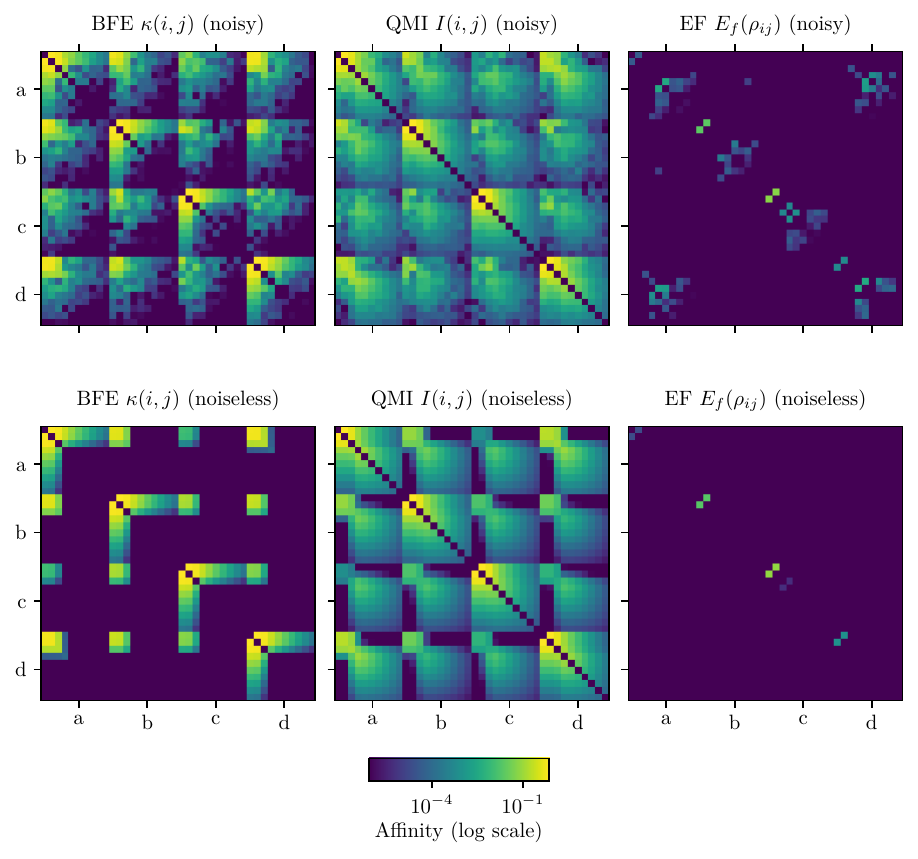}
    \caption{\label{fig:affinities_comparison}\textbf{Comparison of different affinity metrics for a multivariate normal distribution.} Affinity matrices for a multivariate normal distribution ($D=4$, $\eta=1$) in both the Monte-Carlo-sampled noisy case (top row, $N_s=10^4$) and noiseless case computed from a tensor network (bottom row). These are constructed using the Boolean Fourier entropy $\kappa(i,j)$ (left, Equation~\ref{eqn:boolean_fourier_entropy}), the pairwise QMI $I(i,j)$ (middle, Equation~\ref{eqn:pairwise_qmi}), and the entanglement of formation $E_f(\rho_{ij})$ (right, Equation~\ref{eqn:entanglement_formation_defn}).}
\end{figure*}

The appropriate choice of metric for the affinity matrix elements can be highly problem-dependent. In Figure~\ref{fig:affinities_comparison}, we compare these three affinity metrics for the multivariate normal distribution studied previously in Figure~\ref{fig:multivariate_normal_example} (a single realization of the covariance matrix drawn from $\bm{\Sigma}\sim \operatorname{LKJ}_{D=4}(\eta=1)$). For this particular problem instance, we see a great difference in suitability. The Boolean Fourier entropy clearly demonstrates the expected correlation patterns for a highly-correlated smooth function: there are strong correlations between high-significance bits of each variable, decaying correlations in bit significance within the same variable, and extremely weak correlations between low-significance bits of different variables. In contrast, the pairwise QMI indicates moderate correlation between low-significance bits of different variables, possibly due to classical correlation (as opposed to genuine quantum entanglement). For this problem instance, the entanglement of formation is not fit for purpose, yielding exactly zero for most qubit pairs and therefore providing minimal information on the entanglement geometry.

\section{Proofs\label{app:proofs}}

\subsection{Spectral clustering and mean inter-cluster affinity\label{app:spectral_clustering_mean_affinity}}
In Section~\ref{sec:rescaling_affinity_matrices}, we noted that binary spectral clustering approximately minimizes the mean pairwise affinity between elements of each cluster. While this follows straightforwardly from standard proofs, it is not normally presented this way in the literature, so we formalize this here.

\begin{theorem}\label{theorem:scent_bipartition}
    Let $Q$ be the set of qubits to be bipartitioned by an iteration of \textsc{scent} (Algorithm~\ref{alg:scent}) via spectral clustering, into sets $C$ and $\overline{C}$ (such that $C \cap \overline{C} = \emptyset$, $C\cup\overline{C}=Q$). Let $\bm{A}$ be the affinity matrix with elements $A_{ij}$. The bipartition approximately minimizes the mean inter-cluster affinity (i.e. affinities between an element of $C$ and element of $\overline{C}$)
    \begin{equation}
    \frac{1}{|C||\overline{C}|}\sum_{i\in C, j\in \overline{C}}A_{ij}.
    \end{equation}
\end{theorem}
\begin{proof}
Beginning with the mean inter-cluster affinity, simple manipulations show
\begin{align}
    \frac{1}{|C||\overline{C}|}\sum_{i\in C, j\in \overline{C}}A_{ij}&=\frac{1}{2|Q|^2}\left(\frac{|C|+|\overline{C}|}{|C|}+\frac{|C|+|\overline{C}|}{|\overline{C}|}\right) \sum_{i\in C, j\in \overline{C}}A_{ij}\\
    &= \frac{1}{|Q|^2}\left(\frac{|\overline{C}|}{|C|}+\frac{|C|}{|\overline{C}|}+2\right)\sum_{i\in C, j\in \overline{C}}A_{ij} \\
    &= \frac{1}{2|Q|^2}\left(\sum_{i\in C, j\in \overline{C}}A_{ij}\left(\sqrt{\frac{|\overline{C}|}{|C|}}+\sqrt{\frac{|C|}{|\overline{C}|}}\right)^2+\sum_{i\in \overline{C}, j\in C}A_{ij}\left(-\sqrt{\frac{|\overline{C}|}{|C|}}-\sqrt{\frac{|C|}{|\overline{C}|}}\right)^2\right).
\end{align}
Consider now the unnormalized graph Laplacian $\bm{L}\equiv \bm{D}-\bm{W}$, where $\bm{D}$ is the degree matrix and $\bm{W}$ is the weighted adjacency matrix, whose elements are $(\bm{W})_{ij}=A_{ij}$. Defining the vector $\vec{f}$ with elements
\begin{equation}
    (\vec{f})_i=\begin{cases}\sqrt{|\overline{C}|/|C|} & i\in C, \\ -\sqrt{|C|/|\overline{C}|} & i\in \overline{C},\end{cases}
\end{equation}
further manipulations show
\begin{align}
    \frac{1}{|C||\overline{C}|}\sum_{i\in C, j\in \overline{C}}A_{ij} &= \frac{1}{|Q|^2}\sum_{i,j=1}^n w_{ij}((\vec{f})_i-(\vec{f})_j)^2 \\
    &= \vec{f}^T \bm{L} \vec{f} / |Q|^2,
\end{align}
where we have invoked Proposition 1 of Ref.~\cite{von2007tutorial} and followed the manipulations of Section 5.1 therein. It therefore follows that
\begin{equation}
    \operatorname{argmin}_{C\subset Q} \left(\frac{1}{|C||\overline{C}|}\sum_{i\in C, j\in \overline{C}}A_{ij}\right) = \operatorname{argmin}_{C\subset Q}\left(\vec{f}^T \bm{L}\vec{f}\right),
\end{equation}
where the latter is known to be approximately minimized by unnormalized spectral clustering for $k=2$ clusters~\cite{von2007tutorial}, completing the proof. 
\end{proof}
\subsection{Proof of Theorem~\ref{theorem:hilbert_schmidt_distance}\label{app:hilbert_schmidt_distance_proof}}
\begin{proof}
We begin by noting that
\begin{equation}
    ((\widehat{G}_L \otimes \widehat{G}_R)U-\widehat{V})^\dagger ((\widehat{G}_L \otimes \widehat{G}_R)U-\widehat{V})=2-U^\dagger(\widehat{G}_L \otimes \widehat{G}_R)^\dagger \widehat{V} - \widehat{V}^\dagger(\widehat{G}_L \otimes \widehat{G}_R)U, \label{eqn:hs_distance_quadratic_expand}
\end{equation}
where we have used unitarity to simplify $U^\dagger(\widehat{G}_L \otimes \widehat{G}_R)^\dagger (\widehat{G}_L \otimes \widehat{G}_R)U=\mathds{1}$ and $\widehat{V}^\dagger\widehat{V}=\mathds{1}$ (due to the isometric property $V^\dagger V=\mathds{1}$ in the latter case). Applying $\sum_{\alpha=1}^{\chi_{\operatorname{in}}}\widehat{\bra{\psi_\alpha}}\cdot\widehat{\ket{\psi_\alpha}}$ to both sides of Equation~\ref{eqn:hs_distance_quadratic_expand}, we obtain for the right-hand side
\begin{align}
    \operatorname{RHS} &= 2\chi_{\operatorname{in}}-\sum_{\alpha=1}^{\chi_{\operatorname{in}}}\left(\widehat{\bra{\psi_\alpha}}U^\dagger(\widehat{G}_L \otimes \widehat{G}_R)^\dagger \widehat{V}\widehat{\ket{\psi_\alpha}}+\widehat{\bra{\psi_\alpha}}U^\dagger(\widehat{G}_L \otimes \widehat{G}_R)^\dagger \widehat{V}\widehat{\ket{\psi_\alpha}}^*\right) \\
    &=2\chi_{\operatorname{in}}-2\operatorname{Re}\left(\sum_{\alpha=1}^{\chi_{\operatorname{in}}}\widehat{\bra{\psi_\alpha}}\widehat{V}^\dagger(\widehat{G}_L \otimes \widehat{G}_R) U\widehat{\ket{\psi_\alpha}} \right) \\
    &= 2\chi_{\operatorname{in}}\left(1-\mathcal{F}_V(U)\right), \label{eqn:frobenius_norm_rhs}
\end{align}
and for the left-hand side
\begin{align}
\operatorname{LHS}&=\sum_{\alpha=1}^{\chi_{\operatorname{in}}}\widehat{\bra{\psi_\alpha}}((\widehat{G}_L \otimes \widehat{G}_R)U-\widehat{V})^\dagger ((\widehat{G}_L \otimes \widehat{G}_R)U-\widehat{V})\widehat{\ket{\psi_\alpha}} \\
&= \Tr[\mathcal{P}^\dagger ((\widehat{G}_L \otimes \widehat{G}_R)U-\widehat{V})^\dagger ((\widehat{G}_L \otimes \widehat{G}_R)U-\widehat{V}) \mathcal{P}]\\
&= \|(\widehat{G}_L \otimes \widehat{G}_R)U\mathcal{P} - \widehat{V}\mathcal{P}\|^2_{\operatorname{HS}}, \label{eqn:frobenius_norm_lhs}
\end{align}
where $\mathcal{P}$ is the projector onto the subspace spanned by $\widehat{\mathcal{S}}$. Equating \eqref{eqn:frobenius_norm_rhs} and \eqref{eqn:frobenius_norm_lhs}, it follows that maximizing $F_V(U)$ minimizes $\|(\widehat{G}_L \otimes \widehat{G}_R)U\mathcal{P} - \widehat{V}\mathcal{P}\|_{\operatorname{HS}}$, completing the proof.
\end{proof}

\section{TTN compilation algorithm\label{app:TTNcompialtionAlg}}

This appendix provides further details on the TTN compilation method introduced in Section~\ref{sec:compilation_methods} and summarized in Algorithm~\ref{alg:compiler}. The method proceeds as follows:

\begin{algorithmic}[1]
\Procedure{CompileTTN}{$\mathcal{T}=\{T_\nu\}_\nu,k,\delta k,
                        \vec{\varepsilon}$}
    \State $\widetilde{\mathcal{T}} \gets \mathcal{T}$
    \State $\pi \gets$
        \Call{GaugeCompatibleOrder}{$\widetilde{\mathcal{T}}$}

    \ForAll{$\nu\in\pi$}
        \State \textbf{Step 1: Compile up to gauge}
        \State $(C_\nu^{(k)},G_L,G_R) \gets$
            \Call{CompileUpToGauge}{$\widetilde{T}_\nu,k,\varepsilon_1$}
        \State \Call{AbsorbGauges}
            {$\widetilde{\mathcal{T}},\nu,G_L,G_R$}
            \Comment{Push gauges to uncompiled neighbors}

        \State \textbf{Step 2: Reduce gate bodyness}
        \State $C_\nu^{(2)} \gets$
            \Call{LowerGateBodyness}{$C_\nu^{(k)},k,\delta k,\varepsilon_2$}

        \State \textbf{Step 3: Construct gate-reduction family}
        \State $\mathcal{R}_\nu \gets$
            \Call{GateReductionFamily}
                {$C_\nu^{(2)},\widetilde{T}_\nu,\varepsilon_3$}
    \EndFor

    \State $(C_\nu^\star)_\nu \gets$
        \Call{SelectCompilations}
            {$\{\mathcal{R}_\nu\}_\nu,\varepsilon_4$}
    \State \Return
        \Call{AssembleTTNCircuit}{$(C_\nu^\star)_\nu$}
\EndProcedure
\end{algorithmic}

The algorithm starts with an analysis of the TTN connectivity, performed in \textsc{GaugeCompatibleOrder}. Its purpose is to establish the order $\pi$ in which the isometries $T_\nu$ can be compiled. Recall that the isometries are compiled up to gauge transformations that are pushed down the TTN. Thus, the compilation order cannot be arbitrary. However, subject to constraints imposed by the TTN connectivity, some compilations can be performed in parallel.

The main loop of the algorithm goes through all isometries $T_\nu$. Each isometry is generally compiled in three steps. In step 1, the isometry is compiled down to unitaries acting on $k$ qubits using the \textsc{CompileUpToGauge} function, defined as follows.
\begin{algorithmic}[1]
\Function{CompileUpToGauge}{$T,k,\varepsilon$}
    \State $C\gets I$ \Comment{Empty circuit}
    \State $G_L\gets I$ and $G_R\gets I$
    \State $\epsilon \gets$
        \Call{DecompositionError}{$T,C,G_L,G_R$}

    \While{$\epsilon > \varepsilon$}
        \State Add a parameterized $k$-qubit gate to $C$
        \State $(C,G_L,G_R)\gets$
            \Call{OptimizeGatesPositionsAndGauges}
                {$C,T,G_L,G_R$}
        \State $\epsilon \gets$
            \Call{DecompositionError}{$T,C,G_L,G_R$}
    \EndWhile

    \State \Return $(C,G_L,G_R)$
\EndFunction
\end{algorithmic}
Here, $k$ is one of the hyperparameters of the algorithm. It is chosen based on the bond dimension of the input isometry, the time available for compilation, and the observed empirical performance (as measured by the decomposition error as a function of the number of gates; see the discussion around Fig.~\ref{fig:single_isometry_compilation} in the main text). The function \textsc{OptimizeGatesPositionsAndGauges} uses an SVD-based procedure described in Section~\ref{sec:compilation_methods} to update the gates and gauge matrices. Gates are added sequentially to the circuit $C$. The position of each gate is optimized based on its impact on the cost function. The implementation used in this paper considers all possible additions of a 2-qubit gate to the end of the circuit $C$ and chooses the one that yields the greatest reduction in the cost function. Many alternative approaches to growing the circuit may be better suited to particular applications. Such strategies include probabilistic insertion and gradient-magnitude analysis. The choice is dictated primarily by the desired balance between compilation accuracy and the allocated computational resources. The output of \textsc{CompileUpToGauge} function is a quantum circuit built from $k$-qubit gates, together with two gauge matrices, $G_L$ and $G_R$, such that the resulting representation is $\varepsilon$ close to $T$ on the relevant subspace, as defined by the cost function in Eq.~\eqref{eqn:compilation_cost_func}. The gauge matrices are pushed to other isometries $\{ T_{\nu'}\}_{\nu'}$ in \textsc{AbsorbGauges}, according to the connectivity of $\mathcal{T}$.

Step 2 of the algorithm decomposes each gate in the circuit $C$ obtained in step 1 into 2-qubit gates using the \textsc{LowerGateBodyness} function.
\begin{algorithmic}[1]
\Function{LowerGateBodyness}{$C,k,\delta k,\varepsilon$}
    \While{$k \geq 2$}
        \ForAll{each $k$-qubit gate $U$ in $C$}
            \State $\mathcal{D} \gets$
                \Call{CompileGate}{$U,k-\delta k,\varepsilon$}
            \State Replace $U$ in $C$ by the circuit $\mathcal{D}$
        \EndFor
        \State $k \gets k - \delta k$ 
    \EndWhile
    \State \Return $C$
\EndFunction
\end{algorithmic}
The function iteratively decomposes multiqubit gates into smaller gates. It decomposes $k$-qubit gates into $(k-\delta k)$-qubit ones. $\delta k$ is another hyperparameter of the method. We used $\delta k \in \{2,3\}$ to obtain the results discussed in the main text. $\delta k$ may need to be reduced to $1$ if achieving the specified precision $\varepsilon$ proves to be too difficult. The function uses \textsc{CompileGate}, which is similar to the \textsc{CompileUpToGauge} function defined above. The only difference is that it does not include gauge matrices in the circuit optimization:
\begin{algorithmic}[1]
\Function{CompileGate}{$U,r,\varepsilon$}
    \State $C\gets I$ \Comment{Empty circuit}
    \State $\epsilon \gets$ \Call{DecompositionError}{$U,C$}

    \While{$\epsilon > \varepsilon$}
        \State Add a parameterized $r$-qubit gate to $C$
        \State $C\gets$
            \Call{OptimizeGatesAndPositions}{$C,U$}
        \Comment{Jointly optimize parameters and gate positions}
        \State $\epsilon \gets$ \Call{DecompositionError}{$U,C$}
    \EndWhile
    \State \Return $C$
\EndFunction
\end{algorithmic}
The \textsc{CompileUpToGauge} function uses the training set discussed in Section~\ref{sec:compilation_methods}. We apply the same logic to \textsc{CompileGate}. The idea behind the compilation performed here is that the input unitary $U$ needs to be represented accurately only in the context of its position in the circuit. Demanding perfect, unconditioned compilation (as measured by, e.g., a small $||U - U_\mathrm{compiled}||$) will typically lead to circuits with a suboptimal number of gates. The idea of targeted compilation described above is implemented by recomputing the original training data used to compile $T_\nu$ in step 1 with other gates in $C$ that resulted from that compilation.

Step 3 performs gate reduction. It goes through the circuit obtained in step 2 and removes gates while performing additional optimization. The algorithm chooses the gate whose removal causes the smallest increase in the decomposition error. The outcome of step 3 is a family $\mathcal{R}$ of compilations with different gate counts and decomposition errors. Gate reduction is performed by the following function:
\begin{algorithmic}[1]
\Function{GateReductionFamily}{$C,T$}
    \State $C\gets$ \Call{OptimizeGates}{$C,T$}
    \State $\mathcal{R}\gets\emptyset$
    \State Add $C$, its gate count, and its error to $\mathcal{R}$

    \While{another gate can be removed}
        \State Choose the gate $g^\star$ whose optimized removal
            causes the smallest increase in error
        \State $C\gets$
            \Call{OptimizeGates}{$C\setminus\{g^\star\},T$}
        \State Add $C$, its gate count, and its error to $\mathcal{R}$
    \EndWhile

    \State \Return $\mathcal{R}$
\EndFunction
\end{algorithmic}
Here, \textsc{OptimizeGates} is a version of \textsc{OptimizeGatesAndPositions} that performs only SVD-like updates to the gates.

After steps 1--3 are completed for all isometries $T_\nu$, the algorithm moves to its final stage, in which complete circuits approximating the quantum state represented by the TTN $\mathcal{T}$ are constructed. Recall that the compilation of each isometry $T_\nu$ produces a family $\mathcal{R}_\nu$ of circuits characterized by their gate counts and decomposition errors. The algorithm selects one compilation from each family to construct a complete circuit that minimizes the total number of gates while keeping the predicted error below a given threshold. The \textsc{SelectCompilations} function performs this optimization:
\begin{algorithmic}[1]
\Function{SelectCompilations}
    {$\{\mathcal{R}_\nu\}_\nu,\varepsilon$}
    \State Choose
        $(C_\nu^\star,N_\nu^\star,\epsilon_\nu^\star)
        \in\mathcal{R}_\nu$ for every $\nu$,
    \Statex \hspace{\algorithmicindent}
        minimizing $\sum_\nu N_\nu^\star$ subject to
    \Statex \hspace{\algorithmicindent}
        $\epsilon_{\mathrm{total}}((C_\nu^\star)_\nu)
        \leq\varepsilon$.
    \State \Return $(C_\nu^\star)_\nu$
\EndFunction
\end{algorithmic}
The results shown in the main text are obtained by running \textsc{SelectCompilations} with various target errors $\varepsilon$.

\end{document}